\documentclass{lmcs}
\pdfoutput=1

\usepackage{lastpage}
\lmcsdoi{18}{3}{4}
\lmcsheading{}{\pageref{LastPage}}{}{}%
{Jun.~04,~2021}{Jul.~28,~2022}{}


\usepackage[utf8]{inputenc}
\usepackage{amsmath}
\usepackage{amssymb}
\usepackage{amsthm}
  \theoremstyle{theorem}
    \newtheorem{claim}[thm]{Claim}

\usepackage{tikz-cd}
  \usetikzlibrary{arrows,arrows.meta}
  \tikzcdset{arrow style=tikz, diagrams={>=stealth'}}

\usepackage{wasysym}
\usepackage{stmaryrd}
\usepackage{multicol}
\usepackage{booktabs}
\usepackage[mathcal]{euscript}
\usepackage[scr=dutchcal]{mathalfa}
\usepackage{hyperref}
\usepackage{doi}

\setlist[enumerate]{topsep=10pt}

  \DeclareFontFamily{U}{mathc}{}
  \DeclareFontShape{U}{mathc}{m}{it}%
  {<->s*[1.03] mathc10}{}
  \DeclareMathAlphabet{\mathfun}{U}{mathc}{m}{it}
  
\usepackage{mathtools,scalerel}
\newsavebox{\foobox}
\newcommand{\slantbox}[2][.5]{\mbox{%
        \sbox{\foobox}{#2}%
        \hskip\wd\foobox
        \pdfsave
        \pdfsetmatrix{1 0 #1 1}%
        \llap{\usebox{\foobox}}%
        \pdfrestore
}}
\newcommand\rmathcal[2][0mu]{\ThisStyle{\reflectbox{\slantbox[-.55]{%
  $\SavedStyle\mathfun{\mkern#1#2\mkern-#1}$}}}}
  
\newcommand{\mc}[1]{\mathcal{#1}}
\newcommand{\mb}[1]{\mathbb{#1}}
\newcommand{\mf}[1]{\mathfrak{#1}}
\renewcommand{\sf}[1]{\mathsf{#1}}

\newcommand{\topo}[1]{\mathbb{#1}}  % for topological spaces
\newcommand{\cat}[1]{\mathsf{#1}}   % for categories
\newcommand{\fun}[1]{\mathfun{#1}}  % for functors
\renewcommand{\iff}{\quad\text{iff}\quad}

\newcommand{\ov}[1]{\overline{#1}}

\renewcommand{\hat}[1]{\widehat{#1}}
\renewcommand{\phi}{\varphi}
\renewcommand{\epsilon}{\varepsilon}
\renewcommand{\theta}{\vartheta}
\renewcommand{\emptyset}{\varnothing}
\newcommand{\vartri}{{\vartriangle}}
\newcommand{\contra}{\rmathcal[3.5mu]{P}}
\newcommand{\tbigvee}{{\textstyle\bigvee}}

\DeclareMathOperator{\op}{op}

\DeclareMathOperator{\Card}{\mathsf{Card}}

\DeclareMathOperator{\odd}{odd}

\DeclareMathOperator{\Ax}{Ax}

\newcommand{\llb}{\llbracket}
\newcommand{\rrb}{\rrbracket}

\newcommand{\goodbox}{\hspace{.2ex}\text{%
  \tikz[baseline=-.6ex, rounded corners=.01ex, line width=.1ex]
    {\draw (-.6ex,-.6ex) rectangle (.6ex,.6ex);}}\kern.2ex}
\newcommand{\gooddiamond}{\hspace{.2ex}\text{%
  \tikz[baseline=-.6ex, rounded corners=.01ex, rotate=45, line width=.1ex]
    {\draw (-.5ex,-.5ex) rectangle (.5ex,.5ex);}}\kern.2ex}
\renewcommand{\Box}{\goodbox}

\newcommand{\dbox}{\hspace{.2ex}\text{%
  \tikz[baseline=-.6ex, rounded corners=.01ex, line width=.1ex]
    {\draw (-.6ex,-.6ex) rectangle (.6ex,.6ex);
     \draw[fill=black] (0,0) circle(.12ex);}}\kern.2ex}
\newcommand{\ddiamond}{\hspace{.2ex}\text{%
  \tikz[baseline=-.6ex, rounded corners=.01ex, rotate=45, line width=.1ex]
    {\draw (-.5ex,-.5ex) rectangle (.5ex,.5ex);
     \draw[fill=black] (0,0) circle(.12ex);}}\kern.2ex}
\newcommand{\lbox}{\hspace{.2ex}\text{%
  \tikz[baseline=-.6ex, rounded corners=.01ex, line width=.1ex]
    {\draw (-.6ex,-.6ex) rectangle (.6ex,.6ex);
     \draw (-.6ex,0) -- (.6ex,0);}}\kern.2ex}
\newcommand{\ldiamond}{\hspace{.2ex}\text{%
  \tikz[baseline=-.6ex, rounded corners=.01ex, rotate=45, line width=.1ex]
    {\draw (-.5ex,-.5ex) rectangle (.5ex,.5ex);
     \draw (-.5ex,.5ex) -- (.5ex,-.5ex);}}\kern.2ex}

\newcommand{\axitem}[2]{\item[\hypertarget{#1}{$\mathsf{(#2)}$}]}
\newcommand{\axCkap}{\text{\hyperlink{ax:Ckap}{$\mathsf{(C_{\kappa})}$}}}
\newcommand{\axC}{\text{\hyperlink{ax:C}{$\mathsf{(C)}$}}}
\newcommand{\axM}{\text{\hyperlink{ax:M}{$\mathsf{(M)}$}}}
\newcommand{\axN}{\text{\hyperlink{ax:N}{$\mathsf{(N)}$}}}
\newcommand{\axConv}{\text{\hyperlink{ax:conv}{$\mathsf{(Conv)}$}}}
\newcommand{\axCoConv}{\text{\hyperlink{ax:coconv}{$\mathsf{(CoConv)}$}}}
\newcommand{\axCont}{\text{\hyperlink{ax:cont}{$\mathsf{(Cont)}$}}}
\newcommand{\axCent}{\text{\hyperlink{ax:cent}{$\mathsf{(Cent)}$}}}
\newcommand{\axT}{\text{\hyperlink{ax:T}{$\mathsf{(T)}$}}}
\newcommand{\axiv}{\text{\hyperlink{ax:iv}{$\mathsf{(iv)}$}}}
\newcommand{\axfour}{\text{\hyperlink{ax:four}{$\mathsf{(4)}$}}}

\begin{document}

\title[A Coalgebraic Approach to Dualities for Neighborhood Frames]{A Coalgebraic Approach to Dualities for Neighborhood Frames}

\author[G.~Bezhanishvili]{Guram Bezhanishvili}[a]
\address{%Department of Mathematical Sciences\\
         New Mexico State University, Las Cruces, USA}
\email{guram@nmsu.edu}

\author[N.~Bezhanishvili]{Nick Bezhanishvili}[b]
\address{%Institute for Logic, Language and Computation\\
         University of Amsterdam, Amsterdam, The Netherlands}
\email{n.bezhanishvili@uva.nl}

\author[J.~de Groot]{Jim de Groot\lmcsorcid{0000-0003-1375-6758}}[c]
\address{%School of Computing\\
          The Australian National University, Canberra, Australia}
\email{jim.degroot@anu.edu.au}

\begin{abstract}
We develop a uniform coalgebraic approach to J\'{o}nsson-Tarski and Thomason type dualities for various classes of neighborhood frames and neighborhood algebras.
In the first part of the paper we construct an endofunctor on the category of complete and atomic Boolean algebras that is dual to the double powerset functor on $\cat{Set}$. This allows us to show that Thomason duality for neighborhood frames can be viewed as an algebra-coalgebra duality. We generalize this approach to any class of algebras for an endofunctor presented by one-step axioms in the language of infinitary modal logic. As a consequence, we obtain a uniform approach to dualities for various classes of neighborhood frames, including monotone neighborhood frames, pretopological spaces, and topological spaces. 

In the second part of the paper we develop a coalgebraic approach to J\'{o}nsson-Tarski duality for neighborhood algebras and descriptive neighborhood frames. We introduce an analogue of the Vietoris endofunctor on the category of Stone spaces and show that descriptive neighborhood frames are isomorphic to coalgebras for this endofunctor. This allows us to obtain a coalgebraic proof of the duality between descriptive neighborhood frames and neighborhood algebras. Using one-step axioms in the language of finitary modal logic, we restrict this duality to other classes of neighborhood algebras  studied in the literature, including monotone modal algebras and contingency algebras.

We conclude the paper by connecting the two types of dualities via canonical extensions, and discuss when these extensions are functorial. 
\end{abstract}

\maketitle

\tableofcontents

%================================================================================
\section{Introduction}

Categorical dualities linking algebra and topology (such as Stone's seminal duality for Boolean algebras \cite{Sto36} and distributive lattices \cite{Sto38}) have been of fundamental importance in the development of mathematics, logic, and theoretical computer science (see for example~\cite{Joh82,BRV01,com80}). With algebras corresponding to the syntactic side of logical systems and topological spaces to the semantic side, Stone type duality theory provides a powerful mathematical framework for studying various properties of logical systems. That duality theory plays an important r\^{o}le in computer science was emphasized by Plotkin \cite{Plo83} and Smyth \cite{Smy83} who pointed out that the duality between state-transformer and predicate-transformer semantics is an instance of a Stone type duality, as well as by Abamsky \cite{Abr91} who investigated the connection between program logic and domain theory via Stone duality. The latter is an example of a  duality between operational/denotational semantics and program logic. More recently, topological dualities have also been fruitfully explored in understanding minimization of various types of automata \cite{BKP12, Bon14, bezhetal20}, in classification of regular languages \cite{Geh08, Geh16b}, and in analyzing probabilistic systems \cite{Koz13, Fur17}. We refer to \cite{Pan13} 
and \cite{geh16a} for overviews of the use of duality theory in computer science.  Modern applications of duality theory involve algebras enriched with extra operators. 
This can be placed in the context of a duality between algebras for a functor and coalgebras for its dual functor \cite{Ven07}, providing yet another connection with computer science applications as coalgebra is known to be a theory of evolving systems \cite{Rut00, Jac16}. In this paper we  develop a new uniform coalgebraic approach to J\'onsson-Tarski and Thomason type dualities for various classes of neighborhood frames and neighborhood algebras. By doing this we are expanding the toolbox of duality theory for computer science.

J\'onsson-Tarski duality establishes that the category of modal algebras is dually equivalent to the category of descriptive Kripke frames, while Thomason duality establishes a dual equivalence between the category of Kripke frames and the category of complete and atomic modal algebras $(A,\Box)$ such that $\Box$ is completely multiplicative. J\'onsson-Tarski duality generalizes Stone duality between Boolean algebras and Stone spaces (compact Hausdorff zero-dimensional spaces), and Thomason duality generalizes Tarski duality between complete and atomic Boolean algebras and sets.

A coalgebraic proof of J\'onsson-Tarski duality is obtained by lifting Stone duality to an appropriate duality of functors \cite{Abr88,KupKurVen04,KupKurPat05}.
More precisely, let $\cat{Stone}$ be the category of Stone spaces and $\fun{V} : \cat{Stone} \to \cat{Stone}$ the Vietoris endofunctor on $\cat{Stone}$. Then the category $\cat{DKF}$ of descriptive Kripke frames is isomorphic to the category $\cat{Coalg}(\fun{V})$ of coalgebras for $\fun{V}$. In \cite{KupKurVen04} an endofunctor $\fun{K} : \cat{BA} \to \cat{BA}$ on the category $\cat{BA}$ of Boolean algebras is constructed that is dual to $\fun{V}$. Consequently, the category $\cat{Alg}(\fun{K})$ of algebras for $\fun{K}$ is dually equivalent to $\cat{Coalg}(\fun{V})$. Since $\cat{Alg}(\fun{K})$ is isomorphic to the category $\cat{MA}$ of modal algebras, putting the pieces together yields J\'{o}nsson-Tarski duality, stating that $\cat{MA}$ is dually equivalent to $\cat{DKF}$.  

Recently a similar approach was undertaken to prove Thomason duality. It is well known that the category $\cat{KF}$ of Kripke frames is isomorphic to the category $\cat{Coalg}(\fun{P})$ of coalgebras for the covariant powerset endofunctor $\fun{P}$ on the category $\cat{Set}$ of sets. In \cite{BezCarMor20} an endofunctor $\fun{H}$ is constructed on the category $\cat{CABA}$ of complete and atomic Boolean algebras that is dual to $\fun{P}$. Let $\cat{CAMA}$ be the category of complete and atomic modal algebras $(A, \Box)$ such that $\Box$ is completely multiplicative. Then $\cat{CAMA}$ is isomorphic to the category $\cat{Alg}(\fun{H})$ of algebras for $\fun{H}$, and we arrive at Thomason duality.

  In summary, we have two main dualities for normal modal logic, the ``topological'' duality of J\'onsson and Tarski, and the ``discrete'' duality of Thomason. 
  There is a tradeoff in the complexity of the structures involved in these dualities. In J\'onsson-Tarski duality the algebra side of the duality is ``simple'' (finitary), but the frame side is complex (the relational structures need to be equipped with a Stone topology resulting in descriptive Kripke frames). On the other hand, in Thomason duality the frame side of the duality is simple (Kripke frames), while the algebra side is complex (infinitary). 

  J\'onsson-Tarski duality and Thomason duality are related to each other by the following commutative diagram.
  The diagram on the right is a copy of that on the left, but with the modal algebras and frames  presented as algebras and coalgebras for the appropriate  functors.
  The algebraic counterpart of the forgetful functor
  $\fun{U} : \cat{DKF} \to \cat{KF}$ is given by the canonical extension functor
  $\sigma : \cat{MA} \to \cat{CAMA}$ making the diagram on the left commute.
  (We use $\equiv^{\op}$ to indicate dual equivalence.)

  \begin{equation}\label{eq:square-normal}
    \begin{tikzcd}
      \cat{MA}
            \arrow[r, -, "\equiv^{\op}"]
            \arrow[d, "\sigma" swap]
        & \cat{DKF}
            \arrow[d, "\fun{U}"]
        & \cat{Alg}(\fun{K})
            \arrow[r, -, "\equiv^{\op}"]
            \arrow[d, "\sigma" swap]
        & \cat{Coalg}(\fun{V})
            \arrow[d, "\fun{U}"] \\
      \cat{CAMA}
            \arrow[r, -, "\equiv^{\op}"]
        & \cat{KF}
        & \cat{Alg}(\fun{H})
            \arrow[r, -, "\equiv^{\op}"]
        & \cat{Coalg}(\fun{P})
    \end{tikzcd}
  \end{equation}

  We encounter a similar situation when looking at non-normal modal logic---%
  an extension of classical propositional logic with a unary modal operator
  that satisfies only the
  congruence rule. Its algebras are Boolean algebras with an arbitrary endofunction.
  We call these \emph{neighborhood algebras} and write $\cat{NA}$ for
  the category of neighborhood algebras and corresponding homomorphisms.
  The standard geometric semantics of non-normal modal logic is given by \emph{neighborhood frames},
  discovered independently by Scott \cite{Sco70} and Montague
  \cite{Mon70} (see also \cite{Che80,Pac17}). Neighborhood frames can be represented as
  coalgebras for the double contravariant powerset functor on $\mathsf{Set}$ (see, e.g., \cite[Example~9.5]{Ven07}).  
  We write $\cat{NF}$ for the category of neighborhood frames and appropriate
  morphisms.
  Do\v{s}en \cite{Dos89} generalized both J\'onsson-Tarski and
  Thomason dualities to the setting of neighborhood algebras and frames.
  The former gives rise to the notion of \emph{descriptive neighborhood
  frames}, and the latter yields \emph{complete atomic neighborhood algebras}.
  If we write $\cat{DNF}$ and $\cat{CANA}$ for the respective categories, then we obtain the following analogue of Diagram~\eqref{eq:square-normal}:
    \begin{equation}\label{eq:square-non-normal}
    \begin{tikzcd}
      \cat{NA}
            \arrow[r, -, "\equiv^{\op}"]
            \arrow[d, "" swap]
        & \cat{DNF}
            \arrow[d, ""]
            \arrow[d, ""] \\
      \cat{CANA}
            \arrow[r, -, "\equiv^{\op}"]
        & \cat{NF}
    \end{tikzcd}
  \end{equation}
  
  But the analogy is not perfect: in the non-normal case
  there are two natural versions of the canonical extension, the $\sigma$- and $\pi$-extensions, that do not coincide. Moreover, neither is functorial. In fact, $\sigma : \cat{MA} \to \cat{CAMA}$ is a functor because the $\sigma$- and $\pi$-extensions coincide in the presence of normality.
  These two extensions have been investigated 
  in the setting of distributive lattices \cite{GehJon94, GehJon04}, arbitrary lattices \cite{GehHar01}, and even posets \cite{GP08,GJP13}. 
  In Section~\ref{sec:canonical} we will discuss how to obtain functoriality of $\sigma$- and $\pi$-extensions in some special cases.
  
An important subcategory of $\cat{NA}$ is that of \emph{monotone neighborhood frames}. For this subcategory, Hansen \cite{Han03} and Hansen and Kupke \cite{HanKup04} developed an alternative approach to Do\v{s}en duality. In fact, Hansen's descriptive neighborhood frames are $\sigma$-extensions of Do\v{s}en's descriptive neighborhood frames. While Hansen and Kupke mention a coalgebraic approach to Do\v{s}en duality, they do not go as far as to prove Do\v{s}en duality using the coalgebraic approach discussed above.  
In Section~\ref{subsec:ce-fun} we will see how to obtain such a proof from our approach, which we next outline.

  Our main goal is to give a uniform (predominantly coalgebraic)
  approach to dualities for neighborhood frames. We construct an endofunctor
  $\fun{L}$ on the category $\cat{CABA}$ of complete atomic Boolean algebras (CABAs)
  that is dual to the double powerset functor. 
  This endofunctor is obtained by modding out the free CABA over a set by the
  axioms of $\cat{CANA}$. 
  The above construction presumes existence of free CABAs. It is well known that free complete Boolean algebras do not exist \cite{Gai64,Hal64}. 
  However, free CABAs do exist. This follows from the fact that the Eilenberg-Moore algebras of the double contravariant powerset monad are exactly 
CABAs \cite{Tay02}, and that categories of algebras for monads have free objects \cite[Proposition 20.7(2)]{AHS09}. A concrete description of free CABAs was given in \cite{BezCarMor20}.
  We generalize this
  to any subcategory of $\cat{CANA}$
  axiomatized by one-step axioms in the language of infinitary modal logic.
  As a consequence, we obtain a uniform approach to duality theory for classes of neighborhood frames, including monotone
  neighborhood frames, filter frames, and neighborhood contingency frames.
  Additional correspondence results then give rise to dualities for
  pretopological spaces, topological spaces, and their various subcategories.
  
  In the second part of the paper we define an analogue of the Vietoris endofunctor on Stone spaces
  and show that coalgebras for this endofunctor are exactly the descriptive neighborhood
  frames.
  This allows us to obtain a coalgebraic proof of the duality between descriptive
  neighborhood frames and neighborhood algebras. Furthermore, using one-step axioms in the language of finitary modal
  logic, we restrict this duality to other classes of neighborhood algebras studied in the literature such as normal modal algebras and contingency algebras.
  
  This restriction does not always correspond to known dualities.
  For example, when restricting neighborhood algebras to monotone Boolean algebra
  expansions (BAMs) \cite[Chapter 7]{Han03} we do not obtain the duality for
  monotone modal logic of Hansen and Kupke \cite{HanKup04}.
  In fact, the neighborhood frames underlying the descriptive frames in the duality
  obtained from our general theory need not be monotone. 
  The descriptive frames used by Hansen and Kupke can be obtained
  from ours through the theory of canonical extensions.
  In Section \ref{sec:canonical} we provide an axiom which guarantees
  that this construction is functorial. 
  We use this to give an alternative coalgebraic proof of the duality for BAMs
  in \cite{HanKup04}.

%================================================================================
\section{Preliminaries}

%--------------------------------------------------------------------------------
\subsection{Duality theory for normal modal logics}

  Classical modal logic is an extension of  classical propositional logic
  with an additional unary modality $\Box$.
  The modality $\Box$ is called \emph{normal} if it distributes over finite meets.
  By the standard interpretation in a Kripke frame $(X, R)$, a state $x$ in $X$
  satisfies $\Box\phi$ if all its $R$-successors satisfy $\phi$.
  Clearly a state then satisfies $\Box\phi \land \Box\psi$ iff it
  satisfies $\Box(\phi \wedge \psi)$.
  To interpret a non-normal $\Box$ we need to generalize Kripke semantics to
  the so-called neighborhood semantics.
  
\begin{defi}[\cite{Che80,Pac17}]
  A \emph{neighborhood frame} is a pair $(X, N)$
  consisting of a set $X$ and a \emph{neighborhood function}
  $N : X\to\fun{PP}X$, where $\fun{P}X$ is the powerset of $X$.
\end{defi}
  
  For $x\in X$, we call elements of $N(x)$ the \emph{neighborhoods} of $x$.
  A state $x$ in $X$ then satisfies $\Box\phi$ if the set
  $\llb \phi \rrb := \{ y \in X \mid y \Vdash \phi \}$ is a neighborhood of $x$.
  
  Kripke frames can be thought of as special neighborhood frames.
  Indeed, if we identify a Kripke frame $(X, R)$ with the neighborhood frame
  $(X, N_R)$, where $N_R(x) = \{ b \subseteq X \mid R[x] \subseteq b \}$,
  then 
  the interpretation of $\Box\phi$ in $(X, R)$ and $(X, N_R)$ coincides.
  (As usual, $R[x] := \{ y \in X \mid xRy \}$ denotes the set of $R$-successors of $x$.)

  The algebraic semantics of normal modal logic is given by \emph{modal algebras}.
  These are Boolean algebras endowed with a unary function $\dbox:B\to B$ that
  preserves finite meets. Together with $\dbox$-preserving Boolean homomorphisms
  they form the category $\cat{MA}$.
  
  Every Kripke frame $(X,R)$ gives rise to the modal algebra $(\fun{P}X,\dbox_R)$
  where $\dbox_R a := \{ x\in X \mid R[x] \subseteq a\}$.
  By the J\'onsson-Tarski representation theorem \cite{JonTar51}, each modal algebra
  $(B,\dbox)$ can be represented as a subalgebra of the modal algebra
  $(\fun{P}X,\dbox_R)$, where $X$ is the set of ultrafilters of $B$ and $R$ is
  defined on $X$ by $xRy$ iff $\dbox a\in x$ implies $a\in y$ for each $a\in B$.
  In fact, one can endow $X$ with a topology $\tau$ such that the Boolean
  algebra of clopen subsets of $(X, \tau)$ is isomorphic to $B$.
  
  This gives rise to the notion of a \emph{descriptive Kripke frame}, that is,
  a Kripke frame $(\topo{X},R)$ where $\topo{X}$ is a Stone space (a compact
  Hausdorff zero-dimensional space) and $R$ is a continuous relation on $X$;
  see also \cite{Ven07}.
  Let $\cat{KF}$ be the category of Kripke frames and bounded morphisms and
  $\cat{DKF}$  the category of descriptive Kripke frames and continuous
  bounded morphisms. We then have the duality for modal algebras and
  descriptive frames in Theorem \ref{thm:JonTarDual} below.
  The duality on objects traces back to the 1951 paper by
  J\'{o}nsson and Tarski \cite{JonTar51} and was subsequently
  extended to the categorical duality as we know it now \cite{Hal56,Esa74,Gol76a}.

\begin{thm} [J\'onsson-Tarski duality]\label{thm:JonTarDual}
  $\cat{MA}$ is dually equivalent to $\cat{DKF}$.
\end{thm}
  
  This theorem generalizes Stone duality \cite{Sto36} between the category
  $\cat{BA}$ of Boolean algebras and Boolean homomorphisms and the category
  $\cat{Stone}$ of Stone spaces and continuous functions. To obtain a similar
  duality for $\cat{KF}$ we recall Tarski duality for complete and atomic
  Boolean algebras. Let $\cat{CABA}$ be the category of complete atomic
  Boolean algebras and complete Boolean homomorphisms, and let $\cat{Set}$
  be the category of sets and functions. Observe that $\cat{CABA}$ is a
  non-full subcategory of $\cat{BA}$.
  Then Tarski duality states that $\cat{CABA}$ is dually equivalent to
  $\cat{Set}$.
  On objects, this duality stems from Tarski's 1935 paper \cite{Tar35} and
  a statement of the full duality can be found in \cite[Example 4.6(a)]{Joh82}.
  The name ``Tarski duality'' was coined recently in \cite{BezMorOlb19}.
  The duality is obtained by sending a set $X$ to its powerset,
  conceived of as a CABA, and a CABA $A$ to its set of atoms.
  On morphisms, a function $f : X \to X'$ is sent to its inverse image.
  If $h : A \to A'$ is a complete Boolean homomorphism,
  then it has a left adjoint $h^* : A' \to A$, given by
  $h^*(a') = \bigwedge \{ a \in A \mid a' \leq h(a) \}$,
  which restricts to the atoms, and we send $h$ to this restriction.
  Thus, we obtain the contravariant functors $\wp:\cat{Set}\to\cat{CABA}$ and
  $\fun{at}:\cat{CABA}\to\cat{Set}$ which yield Tarski duality.

  Thomason \cite{Tho75} generalized Tarski duality to the category $\cat{CAMA}$,
  whose objects are modal algebras $(A, \dbox)$ such that $A$ is a CABA
  and $\dbox$ is completely multiplicative, i.e., preserves all meets.
  The morphisms of $\cat{CAMA}$ are complete modal algebra homomorphisms.
  
\begin{thm} [Thomason duality]
  $\cat{CAMA}$ is dually equivalent to $\cat{KF}$.
\end{thm}

%--------------------------------------------------------------------------------
\subsection{Duality theory for non-normal modal logics}

  The algebraic semantics of non-normal classical modal logic is given by pairs
  $(B, \dbox)$ where $B$ is a Boolean algebra and $\dbox : B \to B$ is an
  arbitrary (not necessarily meet-preserving) function. Do\v{s}en \cite{Dos89}
  calls these ``modal algebras,'' but to avoid confusion we will call them
  \emph{neighborhood algebras}. We write $\cat{NA}$ for the category of
  neighborhood algebras and $\dbox$-preserving Boolean homomorphisms.
  In addition, we write $\cat{NF}$ for the category of neighborhood frames and
  neighborhood morphisms, where we recall (see, e.g.,
  \cite[Definition 2.9]{Pac17}) that a \emph{neighborhood morphism}
  from $(X, N)$ to $(X', N')$ is a function $f : X \to X'$ such that
  \begin{equation}\label{eq:nbhd-mor}
    a' \in N'(f(x)) \iff f^{-1}(a') \in N(x)
  \end{equation}
  for all $x \in X$ and $a' \subseteq X'$.
  
  Do\v{s}en \cite{Dos89} generalized J\'onsson-Tarski duality and Thomason
  duality to the setting of neighborhood algebras and frames.
  The r\^{o}le of descriptive Kripke frames is now played by descriptive
  neighborhood frames (defined in detail in Section \ref{subsec:dnf-coalg}).
  Together with continuous neighborhood morphisms they comprise the category
  $\cat{DNF}$. Likewise, CAMAs are replaced by \emph{complete atomic neighborhood
  algebras} (CANAs), i.e., neighborhood algebras $(A, \dbox)$ such that $A$ is
  complete and atomic. We write $\cat{CANA}$ for the category of CANAs and
  complete neighborhood algebra homomorphisms. 
  
\begin{thm} [Do\v{s}en \cite{Dos89}]
  \begin{enumerate}
    \item[]
    \item[{\em (1)}] $\cat{NA}$ is dually equivalent to $\cat{DNF}$.
    \item[{\em (2)}] $\cat{CANA}$ is dually equivalent to $\cat{NF}$.
  \end{enumerate}
\end{thm}

%--------------------------------------------------------------------------------
\subsection{The algebra/coalgebra approach}

  Recall \cite[Definition 5.37]{AHS09} that an \emph{algebra} for an endofunctor
  $\fun{L}$ on a category $\cat{A}$ is a pair $(A, \alpha)$ consisting of an
  object $A \in \cat{A}$ and a morphism $\alpha : \fun{L}A \to A$ in $\cat{A}$.
  An \emph{$\fun{L}$-algebra morphism} from $(A, \alpha)$ to $(A', \alpha')$ is
  an $\cat{A}$-morphism $h : A \to A'$ such that the following diagram commutes
  in $\cat{A}$:
  $$
    \begin{tikzcd}
      \fun{L}A
            \arrow[r, "\fun{L}h"]
            \arrow[d, "\alpha" swap]
        & \fun{L}A'
            \arrow[d, "\alpha'"] \\
      A     \arrow[r, "h"]
        & A'
    \end{tikzcd}
  $$
  We write $\cat{Alg}(\fun{L})$ for the category of $\fun{L}$-algebras
  and $\fun{L}$-algebra morphisms.
  
\begin{exa}\label{exm:alg}
  \begin{enumerate}
  \item[]
    \item It is well known that $\cat{MA}$ is isomorphic to the
          category $\cat{Alg}(\fun{K})$ for the endofunctor
          $\fun{K} : \cat{BA} \to \cat{BA}$ that sends a Boolean algebra $B$
          to the free Boolean algebra generated by the meet-semilattice underlying
          $B$ (\cite{Abr88}, \cite[Section 2]{Ghi95},
          \cite[Proposition 3.17]{KupKurVen04}).
    \item \label{it:exm:alg-H}
          As we pointed out in the introduction, free CABAs exist because
          CABAs are exactly the Eilenberg-Moore algebras of the double
          contravariant powerset monad \cite{Tay02} (see also
          \cite[Section  5.1]{bezhetal20}), and categories of Eilenberg-Moore
          algebras have free objects \cite[Proposition 20.7(2)]{AHS09}.
          In \cite{BezCarMor20} a concrete construction of free CABAs was given and it was
          shown that $\cat{CAMA}$ is isomorphic to $\cat{Alg}(\fun{H})$ for the
          endofunctor $\fun{H} : \cat{CABA} \to \cat{CABA}$ that sends a CABA $A$
          to the free CABA generated by the complete meet-semilattice underlying $A$.
          Explicitly, $\fun{H}A$ can be described as the free CABA generated
          by the set underlying $A$, modulo the axioms stating that
          the inclusion map $A \to \fun{H}A$ preserves all meets.
  \end{enumerate}
\end{exa}

  By an easy adaptation of the proof of \cite{Ven07} and \cite{BezCarMor20},
  respectively, one can show that:

\begin{prop}\label{prop:alg} \
  \begin{enumerate}
    \item \label{it:prop:alg-N}
          Let $\fun{N} : \cat{BA} \to \cat{BA}$ be the composition of the
          forgetful functor $\cat{BA} \to \cat{Set}$ and the free functor
          $\cat{Set} \to \cat{BA}$. Then
          $$
            \cat{NA} \cong \cat{Alg}(\fun{N}).
          $$
    \item \label{it:exm:alg-L} 
          Let $\fun{L} : \cat{CABA} \to \cat{CABA}$ be the composition of the
          forgetful functor $\cat{CABA} \to \cat{Set}$ and the free functor
          $\cat{Set} \to \cat{CABA}$. Then
          $$
            \cat{CANA} \cong \cat{Alg}(\fun{L}).
          $$
  \end{enumerate}
\end{prop}
\begin{proof}
  If $B$ is a Boolean algebra, then $\fun{N}B$ is the free Boolean algebra
  generated by the set $\{ \Box b \mid b \in B \}$. So a homomorphism
  $\fun{N}B \to B$ is uniquely determined by its action on elements of the form
  $\Box b$.
  Now given a neighborhood algebra $(B, \dbox)$, define an $\fun{N}$-algebra
  structure $\alpha_{\dbox} : \fun{N}B \to B$ via $\alpha_{\dbox}(\Box b) = \dbox b$.
  Conversely, an $\fun{N}$-algebra $\alpha : \fun{N}B \to B$ gives rise to
  the neighborhood algebra $(B, \dbox_{\alpha})$ where
  $\dbox_{\alpha}b = \alpha(\Box b)$.
  It is easy to verify that these assignments prove the isomorphism of the
  first item on objects. The verification on morphisms is a routine exercise.
  The second item can be proven analogously.
\end{proof}

  The dual notion of an algebra is that of a coalgebra.
  A \emph{coalgebra} for a functor $\fun{T} : \cat{C} \to \cat{C}$
  is a pair $(X, \gamma)$ such that $\gamma : X \to \fun{T}X$ is a morphism
  in $\cat{C}$. A \emph{$\fun{T}$-coalgebra morphism} from $(X, \gamma)$ to
  $(X', \gamma')$ is a $\cat{C}$-morphism $f : X \to X'$ such that the
  following diagram commutes in~$\cat{C}$:
  $$
    \begin{tikzcd}
      X
            \arrow[r, "f"]
            \arrow[d, "\gamma" swap]
        & X'
            \arrow[d, "\gamma'"] \\
      \fun{T}X     \arrow[r, "\fun{T}f"]
        & \fun{T}X'
    \end{tikzcd}
  $$
  We write $\cat{Coalg}(\fun{T})$ for the category of $\fun{T}$-coalgebras
  and $\fun{T}$-coalgebra morphisms.

\begin{exa}\label{exm:coalg}
\begin{enumerate}
  \item[]
    \item \label{it:exm:coalg-P}
          Let $\fun{P} : \cat{Set} \to \cat{Set}$ be the covariant
          powerset functor on $\cat{Set}$. It is well known
          (see, e.g., \cite[Example 2.1]{Rut00})
          that $\cat{KF}$ is isomorphic to $\cat{Coalg}(\fun{P})$.
    \item \label{it:exm:coalg-B}
          Let $\contra : \cat{Set} \to \cat{Set}$ be the contravariant
          powerset functor on $\cat{Set}$.%
            \footnote{Note that $\contra : \cat{Set} \to \cat{Set}$ is the
            composition of $\wp : \cat{Set} \to \cat{CABA}$ with the forgetful
            functor $\cat{CABA} \to \cat{Set}$.}
          Then $\cat{NF}$ is isomorphic to
          $\cat{Coalg}(\contra\contra)$ \cite[Example 9.5]{Ven07}.
          We abbreviate $\fun{B} := \contra\contra$, so
          $\cat{NF} \cong \cat{Coalg}(\fun{B})$.
    \item Let $\fun{V} : \cat{Stone} \to \cat{Stone}$ be the Vietoris endofunctor
          on $\cat{Stone}$. We recall (see e.g.~\cite[Definition 2.5]{KupKurVen04})
          that the \emph{Vietoris space} $\fun{V}\topo{X}$ of a Stone space $\topo{X}$
          is the set of closed subsets of $\topo{X}$
          whose topology is generated by the (clopen) subbasis
          $$
            \lbox a = \{ c \in \fun{V}\topo{X} \mid c \subseteq a \},
            \qquad
            \ldiamond a = \{ c \in \fun{V}\topo{X} \mid c \cap a \neq \emptyset \},
          $$
          where $a$ ranges over the clopen subsets of $\topo{X}$.
          (Note that $\ldiamond a = \fun{V}\topo{X} \setminus \lbox (\topo{X} \setminus a)$,
          so the topology can alternatively be defined by taking the Boolean closure of
          $\{ \lbox a \mid a \in \fun{clp}\topo{X} \}$ as a basis.)
          It is well known that $\fun{V}X$ is a Stone space, and that the assignment
          $\fun{V}$ extends to an endofunctor on $\cat{Stone}$ by setting
          $\fun{V}f : \fun{V}\topo{X} \to \fun{V}\topo{X}' : c \mapsto f[c]$
          for a continuous function $f : X \to X'$.
          It is also well known that $\cat{DKF}$ is isomorphic to $\cat{Coalg}(\fun{V})$;
          see, e.g., \cite{Esa74, Abr88, KupKurVen04}.
    \item In Section \ref{subsec:dnf-coalg} we will define a neighborhood
          analogue $\fun{D}$ of the Vietoris endofunctor such that the category
          $\cat{DNF}$ of descriptive neighborhood frames is isomorphic
          to~$\cat{Coalg}(\fun{D})$.
  \end{enumerate}
\end{exa}

  Since $\cat{MA}$ is isomorphic to $\cat{Alg}(\fun{K})$ and $\cat{DKF}$ is
  isomorphic to $\cat{Coalg}(\fun{V})$, a convenient way to obtain
  J\'onsson-Tarski duality is to prove that Stone duality between $\cat{BA}$
  and $\cat{Stone}$ lifts to a dual equivalence between $\cat{Alg}(\fun{K})$
  and $\cat{Coalg}(\fun{V})$. This can be done by lifting Stone duality to a
  \emph{duality of functors}. 

\begin{defi}
  We call functors $\fun{L} : \cat{BA} \to \cat{BA}$ and
  $\fun{T} : \cat{Stone} \to \cat{Stone}$ \emph{Stone-duals} if the diagram
  below commutes up to natural isomorphism.
  $$
    \begin{tikzcd}
      \cat{BA}
            \arrow[r, shift left=2pt, "\fun{uf}"]
            \arrow[d, "\fun{L}" swap]
        & \cat{Stone}
            \arrow[l, shift left=2pt, "\fun{clp}"]
            \arrow[d, "\fun{T}"] \\
      \cat{BA}
            \arrow[r, shift left=2pt, "\fun{uf}"]
        & \cat{Stone}
            \arrow[l, shift left=2pt, "\fun{clp}"]
    \end{tikzcd}
  $$
  Here $\fun{uf}$ and $\fun{clp}$ are the contravariant functors that
  establish Stone duality.
  That is, $\fun{uf}$ denotes the functor that sends a Boolean algebra to
  its Stone space of ultrafilters and a homomorphism $h$ to $h^{-1}$.
  In the other direction, $\fun{clp}$ is the functor that sends a Stone space to
  its Boolean algebra of clopen sets and a continuous function
  $f$ to $f^{-1}$.
\end{defi}

  Since algebra and coalgebra are dual concepts, we then obtain
  $$
    \cat{Alg}(\fun{L}) \equiv^{\op} \cat{Coalg}(\fun{T}).
  $$
  Thus, J\'{o}nsson-Tarski duality  follows from the fact that
  $\fun{K}$ and $\fun{V}$ are Stone-duals \cite{KupKurVen04}.

  Similarly, since $\cat{CAMA}$ is isomorphic to $\cat{Alg}(\fun{H})$ and
  $\cat{KF}$ is isomorphic to $\cat{Coalg}(\fun{P})$, we can obtain Thomason
  duality by lifting Tarski duality between $\cat{CABA}$ and $\cat{Set}$ to a
  dual equivalence between $\cat{Alg}(\fun{H})$ and $\cat{Coalg}(\fun{P})$. 
  
\begin{defi}\label{def:Tarski-duals}
  We call $\fun{L} : \cat{CABA} \to \cat{CABA}$ and
  $\fun{T} : \cat{Set} \to \cat{Set}$ \emph{Tarski-duals} if the 
  diagram below commutes up to natural isomorphism.
  $$
    \begin{tikzcd}
      \cat{CABA}
            \arrow[r, shift left=2pt, "\fun{at}"]
            \arrow[d, "\fun{L}" swap]
        & \cat{Set}
            \arrow[l, shift left=2pt, "\wp"]
            \arrow[d, "\fun{T}"] \\
      \cat{CABA}
            \arrow[r, shift left=2pt, "\fun{at}"]
        & \cat{Set}
            \arrow[l, shift left=2pt, "\wp"]
    \end{tikzcd}
  $$
  \end{defi}
  As with Stone-dual functors, this gives rise to a duality
  $$
    \cat{Alg}(\fun{L}) \equiv^{\op} \cat{Coalg}(\fun{T}).
  $$
  Thomason duality can now be obtained from the fact that $\fun{H}$ and $\fun{P}$
  are Tarski-duals, which is proven in \cite[Theorem 4.3]{BezCarMor20}.

%================================================================================
\section{Thomason type dualities for neighborhood frames}\label{sec:thomason}

  In this section we derive Thomason type dualities for classes of neighborhood
  frames. We focus on classes of CANAs and classes of neighborhood frames that
  are described by the so-called one-step axioms. As we will see, the
  corresponding dualities are then obtained as algebra/coalgebra dualities.
  Indeed, classes of CANAs defined by one-step axioms can be viewed as
  categories of algebras for some endofunctor on $\cat{CABA}$, while the
  corresponding classes of neighborhood frames as categories of coalgebras for
  an endofunctor on $\cat{Set}$. We can then obtain the desired duality as a
  \emph{duality of functors} (see Definition~\ref{def:Tarski-duals}).
  
  Our results generalize those in \cite{BezCarMor20},
  where a functor duality is proved between the endofunctor $\fun{H}$
  on $\cat{CABA}$ whose algebras are CAMAs (Example~\ref{exm:alg}\eqref{it:exm:alg-H})
  and the powerset functor $\fun{P} : \cat{Set} \to \cat{Set}$
  (Example~\ref{exm:coalg}\eqref{it:exm:coalg-P}).
  In Example~\ref{exm:fun-dual-Cinf} we detail how the results of
  \cite{BezCarMor20} fit in our general scheme. 
  The algebra/coalgebra dualities we obtain pave the way for
  investigations of the resulting classes of frames using methods of
  \emph{coalgebraic logic}, such as developed in e.g.~\cite{Ven07,KupPat11}.

  As running examples, we derive Thomason duality for Kripke frames \cite{Tho75}
  and Do\v{s}en duality for neighborhood frames \cite[Theorem~12]{Dos89}.
  Section \ref{sec:harvest} is dedicated to deriving new dualities
  for various classes of neighborhood frames using
  the theory from this section.

%--------------------------------------------------------------------------------
\subsection{Infinitary languages}\label{subsec:ax}

  We will work in a modal language with arbitrary conjunctions.
  Throughout the paper, for each cardinal $\kappa$ we
  fix a set of variables $V_\kappa$ of cardinality $\kappa$ such that
  $\lambda \leq \kappa$ implies $V_{\lambda} \subseteq V_{\kappa}$
  for all cardinals $\lambda$, $\kappa$.
  Define $\mf{L}^\kappa$ to be the set of formulae generated by the grammar
  $$
    \textstyle
    \phi ::= v \mid \top \mid \neg\phi \mid \bigwedge_{i \in I}\phi_i
  $$
  where $v \in V_\kappa$ and $I$ is some index set of cardinality $< \kappa$.
  We then define $\mf{L}$ as the proper class that contains the formulae in
  $\mf{L}^{\kappa}$ for all cardinals $\kappa$.

  The study of propositional and first-order languages with infinite
  conjunctions was pioneered by Scott and Tarski \cite{ScoTar58} and
  Tarski \cite{Tar58}. In both these references
  the size of conjunctions (and quantifiers) is bounded by some cardinal.
  First-order logics that allow conjunctions and disjunctions over arbitrary
  sets of formulae have been studied comprehensively; see, e.g., \cite{Cha68}, 
  \cite[Part~C]{BarFef85}, \cite{Bel16}, and the references therein.
  
  We next extend the grammar of $\mf{L}$ with a modal operator.
   
\begin{defi}\label{def:inf-lan}
  For each cardinal $\kappa$,
  define $\mf{L}_{\Box}^{\kappa}$ as the set of formulae
  generated by the grammar
  $$
    \phi ::= v
      \mid \top
      \mid \neg\phi
      \mid \bigwedge_{i \in I}\phi_i
      \mid \Box\phi
  $$
  where $v \in V_{\kappa}$ and $I$ is an index set of cardinality $< \kappa$.
  We then define $\mf{L}_{\Box}$ as the proper class that
  contains the formulae in $\mf{L}_{\Box}^{\kappa}$ for all cardinals $\kappa$.
  We will think of elements of $\mf{L}_{\Box}$ as axioms. 
\end{defi}
  
  Note that $\mf{L}^{\kappa}$ is the $\Box$-free fragment of $\mf{L}_{\Box}^{\kappa}$.

\begin{defi}\label{def:ass-sat}
  Let $(A, \dbox)$ be a CANA.
  An \emph{assignment} for $(A, \dbox)$ is a family $\theta_{\kappa} : V_{\kappa} \to A$,
  where $\kappa$ ranges over the cardinals,
  such that $\theta_{\lambda}$ and $\theta_{\kappa}$ agree on all variables in $V_{\lambda}$ whenever $\lambda \leq \kappa$.
  Every assignment $\theta_{\kappa} : V_{\kappa} \to A$ can be extended in an obvious
  way to a map
  $$
    \hat{\theta}_{\kappa} : \mf{L}_{\Box}^\kappa \to A.
  $$
  An assignment gives rise to a map $\hat{\theta} : \mf{L}_{\Box} \to A$
  that sends $\phi \in \mf{L}_{\Box}^{\kappa}$ to $\hat{\theta}_{\kappa}(\phi) \in A$.
  \begin{enumerate}
  \item If $\phi \in \mf{L}_{\Box}$, then we say that \emph{$(A, \dbox)$ validates $\phi$}
        and write $(A, \dbox) \Vdash \phi$ if $\hat{\vartheta}(\phi) = 1$
        for every assignment $\vartheta : (V_{\kappa})_{\kappa \in \Card} \to A$. 
  \item If $\Ax \subseteq \mf{L}_{\Box}$ is a class of axioms, then we say
        that \emph{$(A, \dbox)$ validates $\Ax$}
        and write $(A, \dbox) \Vdash \Ax$ if $(A, \dbox) \Vdash \phi$
        for all $\phi \in \Ax$.
  \end{enumerate}
\end{defi}

  We sometimes write $\theta : V \to A$ for the family of assignments
  $\theta_{\kappa} : V_{\kappa} \to A$. By $v \in V$ we then mean a variable
  in $V_{\kappa}$ for some $\kappa$, and $\theta(v)$ denotes $\theta_{\kappa}(v)$,
  where $\kappa$ is such that $v \in V_{\kappa}$.
  An assignment $\theta : V \to A$ for a CABA $A$
  gives rise to several other maps, listed in Table~\ref{table:theta-maps}.
  Note that we use the font $\mf{L}$ for languages, and
  $\fun{L}$ for the endofunctor on $\cat{CABA}$ from Proposition~\ref{prop:alg}.

\begin{table}[h]
  \centering
  \begin{tabular}{l p{6.7cm}l}
    \toprule
      Map
        & Purpose
        & Location \\ \toprule
      $\hat{\theta} : \mf{L}_{\Box} \to A$
        & To evaluate axioms in a CANA $(A, \dbox)$
        & Def.~\ref{def:ass-sat} \\ \midrule
      $\ov{\theta} : (\mf{L}_{\Box})^1 \to \fun{L}A$
        & To evaluate one-step axioms in $\fun{L}A$
        & Def.~\ref{def:theta-ov} \\ \midrule
      $\theta^t : (\mf{L}_{\Box})^1 \to \fun{P}A$
        & To evaluate one-step axioms as subsets of the powerset of $A$
        & Def.~\ref{def:whattocallthis} \\
    \bottomrule
  \end{tabular}
  \caption{Different maps arising from $\theta$.}
  \label{table:theta-maps}
\end{table}

\begin{defi}
  If $\Ax$ is a class of axioms, then we write $\cat{CANA}(\Ax)$ for the
  full subcategory of $\cat{CANA}$ whose objects validate $\Ax$.
\end{defi}

%--------------------------------------------------------------------------------
\subsection{One-step axioms}

  We next concentrate on the so-called one-step axioms. Intuitively, these are
  formulae in $\mf{L}_{\Box}$ such that every variable occurs in the scope of
  precisely one box. We will see that, in case $\Ax$ consists solely of one-step
  axioms, the category $\cat{CANA}(\Ax)$ is isomorphic to
  $\cat{Alg}(\fun{L}_{\Ax})$ for some endofunctor $\fun{L}_{\Ax}$ on $\cat{CABA}$.
  
\begin{defi}\label{def:one-step}
  For each cardinal $\kappa$,
  define $(\mf{L}_{\Box}^{\kappa})^1$ as the set of
  formulae generated by the grammar
  $$
    \textstyle
    \phi ::= \Box\pi \mid \top \mid \neg\phi \mid \bigwedge_{i \in I}\phi_i
  $$
  where $\pi \in \mf{L}(V_\kappa)$ and $I$ is some index set of cardinality
  $< \kappa$. We then define $(\mf{L}_{\Box})^1$ as the proper class that
  contains the formulae in $(\mf{L}_{\Box}^{\kappa})^1$ for all cardinals $\kappa$.
  A \emph{one-step axiom} is a formula $\phi \in (\mf{L}_{\Box})^1$.
\end{defi}

\begin{rem}
  When viewed as formulae in the modal language $\mf{L}_{\Box}$, the formulae in
  $(\mf{L}_{\Box})^1$ are sometimes referred to as ``formulae of modal depth 1''
  or ``rank 1 formulae.'' This justifies our notation. Observe that every
  one-step axiom is in particular an axiom in the sense of Section~\ref{subsec:ax}.
\end{rem}

\begin{exa}
  \begin{enumerate}
    \item[]
    \item Using the standard abbreviations $\bot, \to, \leftrightarrow$,
          and $\bigvee$, examples of one-step axioms are
          $\Box v \to \Box(v \vee u)$,
          $\Box v \wedge \Box u \leftrightarrow \Box (v \wedge u)$, and
          $$
            \bigwedge_{\lambda < \kappa} \Box v_{\lambda}
              \leftrightarrow \Box \bigg(\bigwedge_{\lambda < \kappa} v_{\lambda} \bigg),
          $$
          where $v, u, v_{\lambda} \in V_\kappa$ and $\lambda,\kappa$ are cardinals.
    \item On the other hand, the axioms $\Box\Box v \to \Box v$, $v \to \Box v$, 
          and $\Box v \wedge v \leftrightarrow \Box\Box v$ are not one-step axioms.
  \end{enumerate}
\end{exa}

  The appeal of one-step axioms lies in the fact that they define endofunctors
  on $\cat{CABA}$ in a structured way. All of these are subfunctors of
  $\fun{L} : \cat{CABA} \to \cat{CABA}$
  from Proposition~\ref{prop:alg}\eqref{it:exm:alg-L}.
  If $\phi$ is a one-step axiom and
  $\theta : V \to A$ is an assignment, then $\theta$ gives rise to a map
  $\ov{\theta} : (\mf{L}_{\Box})^1 \to \fun{L}A$.
  Note that this does \emph{not} rely on any CANA-structure on $A$.
  Indeed, we can use the fact that $\phi$ is of modal depth 1 and define
  $\ov{\theta}$ as follows.
  
\begin{defi}\label{def:theta-ov}
  For a CABA $A$ and assignment $\theta : V \to A$
  we define $\ov{\theta} : (\mf{L}_{\Box})^1 \to \fun{L}A$ recursively via
  \begin{align*}
    \ov{\theta}(\Box\pi) &= \Box\hat{\theta}(\pi), \\
    \intertext{for $\pi \in \mf{L}$ (which is well defined because $\pi$
               does not contain any boxes), and}
    \ov{\theta}(\top) &= 1 \\
    \ov{\theta}(\neg\phi) &= \neg\ov{\theta}(\phi) \\
    \ov{\theta}\Big(\bigwedge_{i \in I}\phi_i\Big) &= \bigwedge_{i \in I} \ov{\theta}(\phi_i)
  \end{align*}
\end{defi}

  Intuitively, the endofunctor on $\cat{CABA}$ corresponding to a collection
  $\Ax$ of one-step axioms sends $A \in \cat{CABA}$ to the free complete atomic Boolean
  algebra generated by $A$ modulo (instantiations of) the axioms,
  i.e., modulo the relations $\ov{\theta}(\phi) = 1$, where $\phi \in \Ax$
  and $\theta$ is an assignment for $A$.
  Before defining this functor, we recall the notion of a
  complete congruence on a CABA.

\begin{defi}
  A \emph{complete congruence} on a CABA $A$ is an equivalence relation $\sim$
  on the underlying set such that $a \sim b$ implies $\neg a \sim \neg b$,
  and $a_i \sim b_i$ for all $i$ in some index set $I$ implies
  $\bigwedge_{i \in I} a_i \sim \bigwedge_{i \in I} b_i$.
\end{defi}
  
  If $\sim$ is a complete congruence on $A$, then 
  we can define the quotient CABA $A/{\sim}$. Writing $[a]$ for the
  equivalence class of $a$, the CABA-operations on $A/{\sim}$ are 
  defined by $1 = [1_A]$, $\neg [a] = [\neg a]$,
  and $\bigwedge [a_i] = [\bigwedge a_i]$.
  Furthermore, the quotient map $q : A \to A/{\sim} : a \mapsto [a]$
  is a complete homomorphism.
  
  The collection of complete congruences on a CABA is closed under arbitrary
  intersections. Therefore, if we have a collection $R$ of equations on $A$ (that is, a collection of equations of the form $a = b$, where $a, b \in A$),
  then we can define the complete congruence generated by $R$ to be the smallest
  complete congruence $\sim_R$ that contains $a \sim_R b$ for all equations
  $a = b$ in $R$.
  Consequently, $A/{\sim_R}$ is the largest quotient (= complete homomorphic image)
  of $A$ in which all equations in $R$ hold.

\begin{defi}
  Let $\Ax$ be a class of one-step axioms. Write
  $\Ax^{\kappa} := \Ax \cap \mf{L}_{\Box}^{\kappa}$ for the subset
  of one-step axioms from $\mf{L}_{\Box}^{\kappa}$
  and $\Ax^{\leq \kappa}$ for the union $\bigcup_{\lambda \leq \kappa} \Ax^{\lambda}$.
  We call $\Ax$ \emph{increasing} if for every CABA $A$
  and assignment $\theta : V \to A$, the following condition holds:
  If $\kappa$ is the cardinality of the powerset of $\fun{L}A$
  and $\ov{\theta}(\phi) = 1$ for all $\phi \in \Ax^{\leq\kappa}$,
  then $\ov{\theta}(\phi) = 1$ for all $\phi \in \Ax$.
\end{defi}

  The idea behind an increasing set of axioms is that, when we want to
  take the quotient of a free CABA with axioms in a proper class $\Ax$,
  we only need to use axioms from $\mf{L}_{\Box}^{\kappa}$,
  where $\kappa$ is the cardinality of the powerset of the free CABA.
  This ensures that the quotient is well defined.
  An example of a class of axioms that is increasing
  is given after Example~\ref{exm:fun-L-empty}.

  If $\Ax$ consists of a set of axioms (rather than a proper class),
  then we can make it increasing by simply adding for each axiom
  $\phi \in \Ax^{\kappa}$, all instantiations
  of $\phi$ where the variables are replaced by variables from
  $V_{\lambda}$ where $\lambda < \kappa$.
  In practice, when dealing with a set of axioms, we will not make this
  explicit.
  
  We are now ready to define an endofunctor $\fun{L}_{\Ax}$ on $\cat{CABA}$
  from an increasing collection $\Ax$ of one-step axioms.

\begin{defi}\label{def:ax-fun}
  Let $\Ax$ be an increasing collection of one-step axioms.  
  For $A \in \cat{CABA}$, define $\fun{L}_{\Ax}A$ to be
  the free complete atomic Boolean algebra generated by the set
  $\{ \Box a \mid a \in A \}$ modulo the congruence ${\sim}_{\Ax}$ generated by
  $\{ \ov{\vartheta}(\phi) \sim_{\Ax} 1 \}$, where $\phi$ ranges over the
  axioms in $\Ax$ and $\vartheta$ ranges over the assignments
  $V \to A$ for the variables in $\phi$.
  For a complete homomorphism $h : A \to A'$ define $\fun{L}_{\Ax}h$ on generators by
  $$
    \fun{L}_{\Ax}h(\Box a) = \Box h(a).
  $$
\end{defi}

  If $\Ax$ consists of a single axiom $\sf{ax}$, then we write
  $\fun{L}_{\sf{(ax)}}$ instead of $\fun{L}_{\{ \sf{(ax)}\}}$.
  Many well-known functors in modal logic can be obtained via
  Definition~\ref{def:ax-fun}. We give two examples.

\begin{exa}\label{exm:fun-L-empty}
  If we work with no axioms (i.e., $\Ax = \varnothing)$, then the
  functor that arises from Definition~\ref{def:ax-fun} is precisely the
  functor $\fun{L}$ from Proposition~\ref{prop:alg}\eqref{it:exm:alg-L}.
\end{exa}

  In our next example we recover the endofunctor
  $\fun{H} : \cat{CABA} \to \cat{CABA}$ from \cite[Section~4]{BezCarMor20},
  where $\fun{H} $ sends a CABA $A$ to the free CABA generated by the 
  complete meet-semilattice underlying $A$
  (see Example~\ref{exm:alg}\eqref{it:exm:alg-H}).
  
  Intuitively, the CABA $\fun{H}A$ is the free CABA generated by the set
  $\{ \Box a \mid a \in A \}$ modulo
  $\bigwedge_{b \in X} \Box b = \Box \bigwedge X$
  for every $X \subseteq A$.
  Although this looks like the instantiation of a single one-step axiom, it is not.
  Indeed, if we define
  $$
    \phi = \bigwedge_{v \in V_{\kappa}} \Box v \leftrightarrow \Box \bigg( \bigwedge V_{\kappa} \bigg)
  $$
  then the axiom $\phi$ only implies \emph{$\kappa$-distributivity}.
  To remedy this, we work with a class of axioms indexed by the class $\Card$
  of cardinal numbers. For each $\kappa \in \Card$, define
  \begin{enumerate}
    \axitem{ax:Ckap}{C_{\kappa}}
          $\bigwedge \{ \Box v \mid v \in V_{\kappa} \}
            \leftrightarrow \Box \bigwedge V_{\kappa}$.
  \end{enumerate}
  Now set
  $$
    \mc{C}_{\infty} = \{ \axCkap \mid \kappa \in \Card \}.
  $$

\begin{exa}\label{exm:fun-L-H}
  Consider the increasing collection of axioms $\Ax = \mc{C}_{\infty}$.
  Then the construction of Definition~\ref{def:ax-fun} yields the
  functor $\fun{H}$ from \cite{BezCarMor20}.
  Its algebras correspond to CAMAs.
\end{exa}

  Thus, incidentally, the previous example also illustrates the need to allow
  a proper class of axioms, rather than just a set.

  Every $\fun{L}_{\Ax}$-algebra $(A, \alpha)$ gives rise to a complete atomic
  algebra $(A, \dbox_{\alpha})$, where
  $$
    \dbox_{\alpha} : A \to A : a \mapsto \alpha(\Box a).
  $$
  Furthermore, if $\phi \in \Ax$, then since
  $\hat{\vartheta}(\phi) \sim_{\Ax} \top$ for all assignments $\vartheta : V \to A$,
  we have $(A, \dbox_{\alpha}) \Vdash \phi$.
  Conversely, if $(A, \dbox)$ is a CANA and $(A, \dbox) \Vdash \Ax$,
  then we can define an $\fun{L}_{\Ax}$-algebra structure map
  $\alpha_{\dbox} : \fun{L}_{\Ax}A \to A$ on generators by
  $\alpha_{\dbox}(\Box a) = \dbox a$.
  The fact that $(A, \dbox) \Vdash \Ax$ implies that $\alpha_{\dbox}$ is
  well defined.
  
  It is easy to see that the two assignments above define a bijection between
  objects of $\cat{Alg}(\fun{L}_{\Ax})$ and objects of $\cat{CANA}(\Ax)$.
  We can extend this to a natural isomorphism in a standard way.
  
\begin{thm}\label{thm:alg-cana}
  If $\Ax$ is an increasing collection of one-step axioms, then 
  $$
    \cat{Alg}(\fun{L}_{\Ax}) \cong \cat{CANA}(\Ax).
  $$
\end{thm}

\begin{proof}[Proof sketch.]
  The isomorphism on objects has already been sketched.
  To prove the isomorphism on morphisms, let $(A, \alpha)$ and
  $(A', \alpha')$ be two $\fun{L}_{\Ax}$-algebras with the corresponding
  CANAs $(A, \dbox)$ and $(A', \dbox')$.
  We claim that a complete homomorphism $h : A \to A'$ is an $\fun{L}_{\Ax}$-algebra
  morphism from $(A, \alpha)$ to $(A', \alpha')$ iff 
  it is a CANA-morphism from $(A, \dbox)$ to $(A', \dbox')$.
  
  Note that $h$ is an $\fun{L}_{\Ax}$-algebra morphism iff
  \begin{equation}\label{eq:Lax-alg-mor}
    h(\alpha(\Box a)) = \alpha'(\fun{L}_{\Ax}h(\Box a))
    \qquad \text{for all $a \in A$.}
  \end{equation}
  Since $\alpha(\Box a) = \dbox a$ and
  $\alpha'(\fun{L}_{\Ax}h(\Box a)) = \alpha'(\Box h(a)) = \Box'h(a)$,
  \eqref{eq:Lax-alg-mor} holds iff
  $$
    h(\dbox a) = \dbox'(h(a)) \qquad \text{for all $a \in A$.}
  $$
  In other words, \eqref{eq:Lax-alg-mor} holds iff $h$ is
  a CANA-morphism.
\end{proof}

  Going back to the class of axioms $\mc{C}_{\infty}$ from Example~\ref{exm:fun-L-H}, we see that Theorem~\ref{thm:alg-cana}
  generalizes \cite[Theorem~4.7]{BezCarMor20}:

\begin{exa}\label{exm:Cinfty-H-CAMA}
  Suppose $\Ax = \mc{C}_{\infty}$.
  Then $\fun{L}_{\Ax} = \fun{H}$ and $\cat{CANA}(\Ax) \cong \cat{CAMA}$,
  so that
  $$
    \cat{Alg}(\fun{H}) \cong \cat{CAMA}.
  $$
\end{exa}

%--------------------------------------------------------------------------------
\subsection{A functor duality theorem}\label{subsec:fdt-set}

  We now classify the atoms of
  $\fun{L}_{\Ax}A$ as subsets of $A$. This then gives rise to an endofunctor on
  $\cat{Set}$ which is a subfunctor of $\fun{B}$ of
  Example~\ref{exm:coalg}\eqref{it:exm:coalg-B} and is dual to $\fun{L}_{\Ax}$.
  We make use of the following notation to characterize the subsets of $A$ we are
  interested in.
  
\begin{defi}
  Let $A$ be a set and $a \in A$. Define
  $$
    \lbox a = \{ W \subseteq A \mid a \in W \}.
  $$
\end{defi}

  We can assign to each one-step axiom a subset of $\fun{P}A$.

\begin{defi}\label{def:whattocallthis}
  Let $\phi \in (\mf{L}^{\kappa}_{\Box})^1$ be a one-step axiom,
  $A$ a set and $\vartheta : V \to A$ an
  assignment.
  Define $\theta^t(\phi)$ to be
  the subset of $\fun{P}A$ given recursively by:
  \begin{align*}
    \theta^t(\Box v) &= \lbox \vartheta(v) \\
    \theta^t(\top) &= \fun{P}A \\
    \theta^t(\neg\phi) &= \fun{P}A \setminus \theta^t(\phi) \\
    \theta^t\textstyle (\bigwedge \phi_i) &= \textstyle\bigcap \{ \theta^t(\phi_i) \}
  \intertext{It then follows that:}
    \theta^t(\tbigvee \phi_i)
      &= \bigcup \{ \theta^t(\phi_i) \} \\
    \theta^t(\phi \to \psi)
      &= \{ W \subseteq A \mid W \in \theta^t(\phi) \Rightarrow W \in \theta^t(\psi) \} \\
    \theta^t(\phi \leftrightarrow \psi)
      &= \{ W \subseteq A \mid W \in \theta^t(\phi) \Leftrightarrow W \in \theta^t(\psi) \}
  \end{align*}
  We say that $W$ is a \emph{$\phi$-subset} of $A$
  if $W \in \theta^t(\phi)$ for every assignment $\vartheta$ of the variables in $V$.
  If $\Ax$ is a collection of axioms, then we say that
  $W$ is an \emph{$\Ax$-subset} if $W$ is a $\phi$-subset for all $\phi \in \Ax$.
\end{defi}

  The next lemma witnesses the significance of $\Ax$-subsets by
  proving a bijective correspondence between atoms
  of $\fun{L}_{\Ax}A$ and $\Ax$-subsets of $A$.
  Recall that atoms of a CABA $A$ correspond bijectively to complete
  homomorphisms into the two-element Boolean algebra $2$: If $a \in A$ is an atom, then
  $p_a : A \to 2$, given by $p_a(b) = 1$ iff $a \leq b$, defines a
  complete homomorphism. Conversely, every complete homomorphism $p$ arises in this way,
  where $a = \bigwedge \{ b \in A \mid p(b) = 1 \}$.

\begin{lem}\label{lem:uf-set}
  Let $A \in \cat{CABA}$ and let $\Ax$ be an increasing collection of axioms.
  Then the atoms of $\fun{L}_{\Ax}A$ correspond bijectively
  to $\Ax$-subsets of $A$.
\end{lem}
\begin{proof}
  We view atoms of $\fun{L}_{\Ax}A$ as complete homomorphisms
  $p : \fun{L}_{\Ax}A \to 2$.
  Since $\fun{L}_{\Ax}A$ is defined by generators and relations,
  $p$ is uniquely determined by its action 
  on the generators of $\fun{L}_{\Ax}A$, i.e., the elements of the
  form $\Box b$ with $b \in A$. Let $W_p \subseteq A$ be the set
  \begin{equation}\label{eq:uf-to-set}
    W_p = \{ b \in A \mid p(\Box b) = 1 \}.
  \end{equation}
  Conversely, given a subset $W \subseteq A$, we define
  $p_W : \fun{L}_{\Ax}A \to 2$ on generators by
  $p_W(\Box b) = 1$ iff $b \in W$.
  
  In order to prove that these assignments are well defined, we use the fact that
  complete homomorphisms $p : \fun{L}_{\Ax}A \to 2$ correspond to complete homomorphisms
  $p' : \fun{L}A \to 2$ whose kernel contains the complete congruence $\sim_{\Ax}$
  generated by (instantiations of) the axioms in $\Ax$ (see Definition~\ref{def:ax-fun}).
  Therefore, for $W \subseteq A$,
  we let $p'_W : \fun{L}A \to 2$ be the complete homomorphism defined
  by $p'_W(\Box b) = 1$ iff $b \in W$.
  
  \begin{claim}
    Let $p : \fun{L}A \to 2$ be a complete homomorphism. 
    Then for all one-step axioms $\phi$ and assignments
    $\vartheta : V \to A$ we have
    $$
      p(\ov{\vartheta}(\phi)) = 1 \iff W_p \in \theta^t(\phi).
    $$
  \end{claim}
  \begin{proof}[Proof of claim]
    We proceed by induction on the complexity of $\phi$.
    If $\phi = \Box v$, where $v$ is in one of the $V_{\kappa}$, then we have
    \begin{align*}
      p(\ov{\vartheta}(\Box v)) = 1
        &\iff p(\Box\vartheta(v)) = 1 \\
        &\iff \vartheta(v) \in W_p \\
        &\iff W_p \in \lbox(\vartheta(v)) \\
        &\iff W_p \in \theta^t(\Box v).
    \end{align*}
    Let $\phi = \top$. By definition, $p(\ov{\vartheta}(\phi)) = 1$ for all
    complete homomorphisms $p$. Since $\vartheta^t(\phi) = \fun{P}A$, the result
    holds for $\phi = \top$. For negation, we have
    \begin{align*}
      p(\ov{\vartheta}(\neg\phi)) = 1
        &\iff p(\neg\ov{\vartheta}(\phi)) = 1 \\
        &\iff p(\ov{\vartheta}(\phi)) = 0 \\
        &\iff W_p \notin \theta^t(\phi) &\text{(inductive hypothesis)}\\
        &\iff W_p \in \fun{P}A \setminus \theta^t(\phi) = \theta^t(\neg\phi).
    \end{align*}
    Finally, if $\phi = \bigwedge\phi_i$ then
    \begin{align*}
      \textstyle
      p(\ov{\vartheta}(\bigwedge \phi_i)) = 1
        &\iff p(\textstyle\bigwedge \ov{\vartheta}(\phi_i)) = 1 \\
        &\iff \textstyle\bigwedge p(\ov{\vartheta}(\phi_i)) = 1 \\
        &\iff W_p \in \theta^t(\phi_i) \quad\text{ for all $i$} &\text{(inductive hypothesis)}\\
        &\iff W_p \in \bigcap \theta^t(\phi_i) = \theta^t(\textstyle\bigwedge \phi_i).
    \end{align*}
    This completes the proof of the claim.
  \end{proof}
  
  Now let $p : \fun{L}_{\Ax}A \to 2$ be a complete homomorphism.
  Then composing the quotient map $q : \fun{L}A \to \fun{L}_{\Ax}A$ with $p$
  yields a complete homomorphism
  \begin{equation}\label{eq:p-after-q}
    p \circ q : \fun{L}A \to 2
  \end{equation}
  whose kernel contains $\sim_{\Ax}$.
  Moreover, $W_p = W_{p \circ q}$, where $W_{p \circ q}$ is defined as in
  \eqref{eq:uf-to-set} for the complete homomorphism from \eqref{eq:p-after-q},
  because $q$ sends the generator $\Box a$ of $\fun{L}A$ to the equivalence class
  of $\Box a$ in $\fun{L}_{\Ax}A$.
  Consequently,
  $$
    W_p \in \theta^t(\phi)
      \iff W_{p \circ q} \in \theta^t(\phi)
      \iff (p \circ q)(\ov{\vartheta}(\phi)) = 1
      \iff p(\ov{\vartheta}(\phi)) = 1.
  $$
  Since $\ov{\vartheta}(\phi) \sim_{\Ax} 1$
  for all $\phi \in \Ax$, we have $p(\ov{\vartheta}(\phi)) = 1$ for all $\phi \in \Ax$.
  Using the claim, this proves that $W_p$ is an $\Ax$-subset.
  
  Conversely, a similar computation shows that whenever $W$ is
  an $\Ax$-subset, then the kernel of $p_W' : \fun{L}A \to 2$
  contains $\sim_{\Ax}$, and hence $p_W'$ defines a complete homomorphism
  $p_W : \fun{L}_{\Ax}A \to 2$.
  In addition, 
  for each $b \in A$ we have $p_{W_p}(\Box b) = 1$ iff $b \in W_p$ iff $p(\Box b) = 1$,
  so $p_{W_p} = p$. Similarly, $b \in W_{p_W}$ iff $p_W(\Box b) = 1$ iff $b \in W$,
  and hence $W_{p_W} = W$.
  Thus, these assignments define a bijection.
\end{proof}

  Guided by Lemma~\ref{lem:uf-set}, we define an endofunctor on $\cat{Set}$
  which we then prove to be the Tarski-dual of $\fun{L}_{\Ax}$.

\begin{defi}\label{def:fun-B}
  Let $X$ be a set and $\Ax$ an increasing collection of one-step axioms.
  Define $\fun{B}_{\Ax}X$ to be the set of $\Ax$-subsets of
  $\fun{P}X$.
  For a function $f : X \to X'$ in $\cat{Set}$, define $\fun{B}_{\Ax}f$ by
  $$
    \fun{B}_{\Ax}f
      : \fun{B}_{\Ax}X \to \fun{B}_{\Ax}X'
      : W \mapsto \{ a' \in \fun{P}X' \mid f^{-1}(a') \in W \}.
  $$
\end{defi}

  If $\Ax$ consists of a single axiom $\sf{ax}$, then we write
  $\fun{B}_{\sf{(ax)}}$ instead of $\fun{B}_{\{ \sf{(ax)}\}}$.

\begin{prop}\label{prop:Bax-welldef}
  The assignment $\fun{B}_{\Ax}$ is a well-defined endofunctor on $\cat{Set}$.
\end{prop}
\begin{proof}
  Clearly $\fun{B}_{\Ax}X$ is a set for every set $X$.
  Let $f : X \to X'$ be a function and let $W \in \fun{B}_{\Ax}X$.
  We need to show that $\fun{B}_{\Ax}f(W)$ is in $\fun{B}_{\Ax}X'$,
  that is, $\fun{B}_{\Ax}f(W)$ is an $\Ax$-subset of $\fun{P}X'$.
  But this follows from the fact that $\fun{B}_{\Ax}f(W) \in \lbox a'$
  iff $W \in \lbox f^{-1}(a')$.
  
  Functoriality of $\fun{B}_{\Ax}$ follows from the fact that $\fun{B}_{\Ax}$ is a subfunctor
  of $\fun{B}$.
\end{proof}

  We are ready to prove the main result of this section. 

\begin{thm}[Tarski Functor Duality Theorem]\label{thm:fun-dual}
  The functors $\fun{L}_{\Ax}$ and $\fun{B}_{\Ax}$ are Tarski-duals.
\end{thm}
\begin{proof}
  Define $\xi_X : \fun{at}(\fun{L}_{\Ax}(\wp X)) \to \fun{B}_{\Ax}X$ 
  by sending an atom $a$ corresponding to a complete homomorphism $p$ to $W_p$.
  %$p \mapsto W_p$.
  This defines an isomorphism on objects by
  Lemma \ref{lem:uf-set}. We prove that the assignment
  $\xi = (\xi_X)_{X \in \cat{Set}}
    : \fun{at} \circ \fun{L}_{\Ax} \circ \wp \to \fun{B}_{\Ax}$
  is natural by showing that
  for every function $f : X \to X'$ the diagram
  $$
    \begin{tikzcd}
      \fun{at}(\fun{L}_{\Ax}(\wp X)) 
            \arrow[d, "\fun{at}(\fun{L}_{\Ax}(\wp f))" swap]
            \arrow[r, "\xi_X"]
        & [1em]
          \fun{B}_{\Ax}X
            \arrow[d, "\fun{B}_{\Ax}f"] \\ [1em]
      \fun{at}(\fun{L}_{\Ax}(\wp X'))
            \arrow[r, "\xi_{X'}"]
        & \fun{B}_{\Ax}X'
    \end{tikzcd}
  $$
  commutes.
  To this end, let $u \in \fun{at}(\fun{L}_{\Ax}(\wp X))$ and let
  $p : \fun{L}_{\Ax}(\wp X) \to 2$ be the corresponding complete homomorphism.
  Furthermore, let $a' \subseteq X'$.
  Then
  \begin{align*}
    a' \in \fun{B}_{\Ax}f \circ \xi_X(p)
      &\iff \xi_X(p)(f^{-1}(a')) = 1 \\
      &\iff p(\Box f^{-1}(a')) = 1\\
      &\iff p(\fun{L}_{\Ax}(\wp f)(\Box a')) = 1 \\
      &\iff \fun{at}(\fun{L}_{\Ax}(\wp f))(p)(\Box a') = 1 \\
      &\iff a' \in \xi_{X'} \circ \fun{at}(\fun{L}_{\Ax}(\wp f))(p).
  \end{align*}
  This proves the theorem.
\end{proof}

\begin{cor}\label{cor:fun-dual}
  For every increasing collection $\Ax$ of one-step axioms, we have
  $$
    \cat{Alg}(\fun{L}_{\Ax}) \equiv^{\op} \cat{Coalg}(\fun{B}_{\Ax}).
  $$
\end{cor}

\begin{exa}\label{exm:fun-dual-no-axioms}
  As the notation suggests, the functor $\fun{B}_{\varnothing}$
  defined as in Definition~\ref{def:fun-B} using the empty set of axioms
  is precisely the functor $\fun{B}$ from Example~\ref{exm:coalg}\eqref{it:exm:coalg-B}.
  We have seen that its coalgebras are neighborhood frames.
  Combining this with Example~\ref{exm:fun-L-empty} 
  and Theorem~\ref{thm:fun-dual} yields that $\fun{L}$ and
  $\fun{B}$ are Tarski-duals.
  As a consequence of Corollary~\ref{cor:fun-dual},
  we derive Do\v{s}en's duality \cite{Dos89} for neighborhood frames via
  $$
    \cat{CANA}
      \cong \cat{Alg}(\fun{L})
      \equiv^{\op} \cat{Coalg}(\fun{B})
      \cong \cat{NF}.
  $$
\end{exa}

\begin{exa}\label{exm:fun-dual-Cinf}
  Let $\Ax = \mc{C}_{\infty}$ as in Example \ref{exm:Cinfty-H-CAMA}.
  For a set $X$, the set $\fun{B}_{\Ax}X$ consists of all collections of
  neighborhoods that are upward closed under inclusion and closed under
  arbitrary intersections. Such a collection is uniquely determined by
  its smallest neighborhood, and a straightforward verification shows that
  $\fun{B}_{\Ax}$ is naturally isomorphic to the covariant powerset
  functor $\fun{P} : \cat{Set} \to \cat{Set}$. 
  
  As a consequence of Theorem \ref{thm:fun-dual}, we obtain that $\fun{P}$
  is dual to $\fun{H}$, a result that was recently established
  in \cite[Theorem 4.3]{BezCarMor20}.
  Using the well-known fact that $\fun{P}$-coalgebras are Kripke frames
  and the observation that $\fun{H}$-algebras are CAMAs, we arrive at Thomason
  duality:
  $$
    \cat{CAMA}
      \cong \cat{Alg}(\fun{H})
      \equiv^{\op} \cat{Coalg}(\fun{P})
      \cong \cat{KF}.
  $$
\end{exa}

%================================================================================
\section{Applications}\label{sec:harvest}

  In this section we first derive dualities for various types of neighborhood frames
  using only one-step axioms and Corollary~\ref{cor:fun-dual}.
  This gives rise to Thomason type dualities for monotone neighborhood frames,
  contingency neighborhood frames, and filter frames.
  Each of these is an algebra/coalgebra duality.
  
  Next we show how some of these restrict when we
  invoke further axioms, which are not necessarily one-step axioms.
  These results can be seen as correspondence
  results. Most notably, they allow us to obtain McKinsey-Tarski duality for
  topological spaces (with interior maps) as an easy restriction of the duality for
  filter frames.
  An overview of the dual equivalences discussed in this section
  is given in Table \ref{table:overview}.

%--------------------------------------------------------------------------------
\subsection{One-step dualities}

%................................................................................
\subsubsection{Monotone neighborhood frames}\label{subsubsec:monotone}
  
  Monotone modal logic is a well-studied branch of
  modal logic (see, e.g., \cite{Che80,Han03,HanKup04}).
  The standard semantics for monotone modal logic is given by
  \emph{monotone neighborhood frames}.
  Recall that these are neighborhood frames $(X, N)$ such that $N(x)$
  is upward closed under inclusion as a subset of $\fun{P}X$.
  We write $\cat{MF}$ for the full subcategory of $\cat{NF}$ whose objects are
  monotone neighborhood frames.
  It is well known that $\cat{MF} \cong \cat{Coalg}(\fun{UpP})$, where
  $\fun{UpP} : \cat{Set} \to \cat{Set}$ takes a set $X$ to the collection of
  subsets of $\fun{P}X$ that are upward closed under inclusion \cite{Han03,HanKup04}.
  
  The algebraic semantics of monotone modal
  logic is given by \emph{monotone Boolean algebra expansions} (BAMs for short)
  \cite[Section 7]{Han03}.
  A BAM is a neighborhood algebra $(A, \dbox)$ such that
  $\dbox : A \to A$ is a monotone function (that is, $a\le b$ implies $\dbox a\le\dbox b$).
  Let $\cat{BAM}$ be the full subcategory of $\cat{NA}$ whose objects are BAMs.

\begin{table}
  \centering
  \begin{tabular}{llll}
    \toprule
    Axioms & Algebras & Objects & Location \\ 
    & Frames && \\
    \toprule
    None
      & $\cat{CANA}$ & Complete atomic neighborhood algebras & Example \ref{exm:fun-dual-no-axioms} \\
      & $\cat{NF}$ & Neighborhood frames & \\ \midrule
    \axCont\
      & $\cat{CAContA}$ & Complete atomic contingency algebras & Section \ref{subsubsec:cont} \\
      & $\cat{ContF}$ & Contingency frames & \\ \midrule
    \axConv\
      & $\cat{CACA}$ & Complete atomic convex algebras & Section \ref{subsubsec:convex} \\
      & $\cat{CNF}$ & Convex neighborhood frames & \\ \midrule
    \axM
      & $\cat{CABAM}$ & Complete atomic monotone BA expansions & Section \ref{subsubsec:monotone} \\
      & $\cat{MF}$    & Monotone neighborhood frames & \\ \midrule
    \axN, \axC
      & $\cat{caMA}$ & Modal algebras over CABAs & Section \ref{subsubsec:filter} \\
      & $\cat{FF}$ & Filter frames & \\ \midrule
    \axCkap
      && $\kappa$-additive complete atomic modal algebras & Remark \ref{rem:tanaka} \\
      && $\kappa$-complete neighborhood frames & \\ \midrule
    $\mc{C}_{\infty}$
      & $\cat{CAMA}$ & CAMAs & Example \ref{exm:fun-dual-Cinf} \\
      & $\cat{KF}$ & Kripke frames & \\ \midrule
    \axN, \axC,
      & $\cat{PreTop_{int}}$ & Pretopological spaces & Section \ref{subsec:corr} \\
     \axT
      & $\cat{PreInt}$ & Complete atomic pre-interior algebras & \\ \midrule
    \axN, \axC,
      & $\cat{Top_{int}}$ & Topological spaces & Section \ref{subsec:corr} \\
     \axT, \axfour
      & $\cat{Int}$ & Complete atomic interior algebras & \\ 
    \bottomrule
  \end{tabular}
  \caption{Overview of pairs of dual categories.}
  \label{table:overview}
\end{table}
  
\begin{defi}
  Let $\cat{CABAM}$ be the full subcategory of $\cat{CANA}$ whose objects
  are also BAMs.
\end{defi}

  We can view $\cat{CABAM}$ as a category of algebras for an endofunctor on $\cat{CABA}$.
  To see this, consider the one-step axiom
  \begin{enumerate}[($\mathsf{M}$)]
    \axitem{ax:M}{M}
          $\Box(u \wedge v) \to \Box u$
  \end{enumerate}
  expressing monotonicity.
  As a consequence of Theorem~\ref{thm:alg-cana}, we have:
  
  \begin{cor}\label{cor:CABAM-AlgLM}
  $
    \cat{CABAM} \cong \cat{Alg}(\fun{L}_{\axM}).
  $
  \end{cor}

  For a set $X$, the $\axM$-subsets of
  $\fun{P}X$ are precisely the ones that are up-closed under inclusion.
  It then follows immediately from the definitions that
  $\fun{B}_{\axM}$ defined as in Definition~\ref{def:fun-B} coincides with
  $\fun{UpP} : \cat{Set} \to \cat{Set}$ from \cite[Section 3.1]{HanKup04}.
  As a consequence of Corollary~\ref{cor:fun-dual}, we obtain:

\begin{thm}\label{thm:dual-MF-CABAM}
  The category $\cat{MF}$ of monotone neighborhood frames
  is dually equivalent to $\cat{CABAM}$.
\end{thm}

\begin{proof}
  As a consequence of Theorem \ref{thm:fun-dual}, the functors
  $\fun{L}_{\axM}$ and $\fun{B}_{\axM} = \fun{UpP}$ are dual.
  Moreover, $\cat{CABAM} \cong \cat{Alg}(\fun{L}_{\axM})$ by Corollary \ref{cor:CABAM-AlgLM},
  and $\cat{MF} \cong \cat{Coalg}(\fun{UpP})$ by \cite[Lemma 3.4]{HanKup04}.
  Therefore, Corollary~\ref{cor:fun-dual} implies that
  $\cat{Alg}(\fun{L}_{\axM}) \equiv^{\op} \cat{Coalg}(\fun{UpP})$,
  which proves the theorem.
\end{proof}

%................................................................................
\subsubsection{Duality for neighborhood contingency logic}\label{subsubsec:cont}
  
  A formula is called \emph{contingent} if it is possibly true and possibly
  false. Otherwise it is non-contingent, i.e., it is necessarily true or
  necessarily false. Neighborhood contingency logic was recently introduced in
  \cite{FanDit15} to reason about contingent formulae, and is investigated
  further in \cite{BakDitHan17,Fan18}. The non-contingency modality $\vartri$ is
  interpreted in a neighborhood frame $(X, N)$ by
  $$
    x \Vdash \vartri \phi
      \iff \llb \phi \rrb \in N(x) \text{ or } X \setminus \llb \phi \rrb \in N(x).
  $$
  As a consequence of this definition, the interpretation of formulae
  does not distinguish whether $a \in N(x)$ or $X \setminus a \in N(x)$ or both.
  Therefore, it suffices to only consider the neighborhood
  frames $(X, N)$ that satisfy: for all $x \in X$ and $a \subseteq X$,
  $$
    a \in N(x) \iff X \setminus a \in N(x).
  $$
  We call such frames \emph{contingency frames} and write $\cat{ContF}$
  for the full subcategory of $\cat{NF}$ consisting of contingency frames.
  Then $\cat{ContF}$ is isomorphic to the category of $\fun{B}_{\axCont}$-coalgebras,
  where \axCont\ is the axiom
  \begin{enumerate}
    \axitem{ax:cont}{Cont}
          $\Box v \leftrightarrow \Box \neg v$.
  \end{enumerate}
  Corollary~\ref{cor:fun-dual} implies that
  $\cat{Coalg}(\fun{B}_{\axCont}) \equiv^{\op} \cat{Alg}(\fun{L}_{\axCont})$.
  As a consequence of Theorem~\ref{thm:alg-cana} we can describe the latter
  category of algebras explicitly as the full subcategory of $\cat{CANA}$
  whose objects are the CANAs $(A, \dbox)$ satisfying $\dbox a = \dbox \neg a$ for all $a \in A$.
  We call these \emph{complete atomic contingence algebras}
  and denote the category they form by $\cat{CAContA}$.
  Thus, putting the above together, we obtain:
  
\begin{thm}\label{thm:dual-cont}
  $
    \cat{ContF} \equiv^{\op} \cat{CAContA}.
  $
\end{thm}

%................................................................................
\subsubsection{Convex frames}\label{subsubsec:convex}

  Our next example concerns \emph{convex neighborhood frames}. These are
  neighborhood frames $(X, N)$ such that $N(x)$ is a convex subset of $\fun{P}X$,
  meaning that for all $x \in X$, if $a, a' \in N(x)$ and
  $a \subseteq b \subseteq a'$, then $b \in N(x)$. Write $\cat{CNF}$ for the full
  subcategory of $\cat{NF}$ whose objects are convex neighborhood frames. 
  
  As we will see in Section \ref{sec:canonical}, convexity is closely related to
  the question of functoriality of canonical extensions of neighborhood frames.
  Convexity is captured by the following axiom:
  \begin{enumerate}
    \axitem{ax:conv}{Conv}
          $\Box(v \wedge v') \wedge \Box(v \vee v'') \to \Box v$
  \end{enumerate}
  Therefore,
  $
    \cat{CNF} \cong \cat{Coalg}(\fun{B}_{\axConv}).
  $
  We call the corresponding algebras
  \emph{complete atomic convex algebras},
  and denote by $\cat{CACA}$ the full subcategory of $\cat{CANA}$ whose objects are
  complete atomic convex algebras.
  Then $\cat{CACA} \cong \cat{Alg}(\fun{L}_{\axConv})$,
  and as a consequence of Corollary~\ref{cor:fun-dual}, we obtain:

\begin{thm}\label{thm:dual-CNF-CACA}
  $\cat{CNF} \equiv^{\op} \cat{CACA}$.
\end{thm}

%................................................................................
\subsubsection{Filter frames}\label{subsubsec:filter}
  
  Finally, we consider filter frames \cite[Section 7.2]{Che80}.
  These are of interest because they
  are as close as we can get to topological spaces using only one-step axioms (see below).
  Recall that a \emph{filter} on a set $X$ is a subset $F \subseteq \fun{P}X$ that is
  closed under finite intersections and upward closed under inclusion.
  In particular, this implies that $X \in F$ as $X$ is the empty intersection
  of subsets of $X$.
  A \emph{filter frame} is a neighborhood frame $(X, N)$ such that
  $N(x)$ is a filter on $X$ for each $x \in X$. Let $\cat{FF}$ be
  the full subcategory of $\cat{NF}$ consisting of filter frames.
  
  From the coalgebraic point of view, filter frames are
  $\fun{B}_{\Ax}$-coalgebras, where $\Ax$ consists of the axioms
  \begin{enumerate}
    \axitem{ax:N}{N} $\Box\top$
    \item[$(\mathsf{C})$] \label{it:C}
          $\Box u \wedge \Box v \leftrightarrow \Box(u \wedge v)$
  \end{enumerate}
  We write $\fun{B}_{\wedge}$ for $\fun{B}_{\{\axN, \axC \}}$.
  We then have $\cat{FF} \cong \cat{Coalg}(\fun{B}_{\wedge})$.
  As a consequence of Corollary~\ref{cor:fun-dual}, we obtain the dual equivalence
  \begin{equation}\label{eq:alg-coalg-fi}
    \cat{Coalg}(\fun{B}_{\wedge}) \equiv^{\op} \cat{Alg}(\fun{L}_{\wedge}),
  \end{equation}
  where $\fun{L}_{\wedge}$ abbreviates $\fun{L}_{\{ \axN, \axC \}}$.

  Since \axN\ and \axC\ are the axioms that on the algebra side define modal algebras,
  $\fun{L}_{\wedge}$-algebras are simply modal algebras whose underlying
  Boolean algebra is complete and atomic.
  We write $\cat{caMA}$ for the full subcategory of $\cat{CANA}$ whose
  objects are modal algebras based on complete atomic Boolean algebras.
  By contrast, recall that $\cat{CAMA}$ denotes the full subcategory of
  $\cat{CANA}$ whose objects $(A, \dbox)$ are such that 
  $A$ is complete atomic and $\dbox$ preserves arbitrary meets.
  Therefore, $\cat{CAMA}$ is a full subcategory of $\cat{caMA}$. 
  Rephrasing \eqref{eq:alg-coalg-fi} yields the following generalization of Thomason duality:
  
\begin{thm}\label{thm:dual-FF-FIMA}
    $\cat{FF} \equiv^{\op} \cat{caMA}$.
\end{thm}
  
\begin{rem}\label{rem:tanaka}
  If we require the collection of neighborhoods at each state to be closed
  under intersections of size $< \kappa$, where $\kappa$ is some fixed cardinal,
  then we obtain the \emph{$\kappa$-complete neighborhood frames} from
  \cite[Section 4]{Tan21}. This corresponds to the axiom \axCkap,
  and yields a dual equivalence with \emph{$\kappa$-additive complete atomic modal
  algebras}.
  This category lies in between $\cat{CAMA}$ and $\cat{caMA}$.
\end{rem}

%--------------------------------------------------------------------------------
\subsection{Restrictions/correspondence results}\label{subsec:corr}
  
  If a collection of frames or algebras is not given by one-step axioms, we can
  still derive dualities for them from correspondence
  results for the axioms under consideration, but they are no longer
  algebra/coalgebra dualities. 
  This allows us to derive other interesting dualities, such as
  McKinsey-Tarski duality for topological spaces (and interior maps).
  
  Topological spaces are obtained from filter frames by stipulating
  the reflexivity and transitivity axioms \axT\ and \axfour.
  Adding only \axT\ results in the more general notion of pre-topological spaces.
  For a neighborhood frame $(X,N)$ and $a \subseteq X$, let 
  $$
    \dbox_N(a) = \{ y \in X \mid a \in N(y) \}. 
  $$

\begin{defi}
  \begin{enumerate}
    \item[]
    \item A \emph{pre-topological space} is a filter frame that satisfies
          \begin{enumerate}
            \axitem{ax:cent}{Cent}
                    $a \in N(x)$ implies $x \in a$.
          \end{enumerate}
    \item A \emph{topological space} is a pre-topological space that satisfies
          \begin{enumerate}
            \axitem{ax:iv}{iv}
                    $a \in N(x)$ implies $\dbox_N(a) \in N(x)$.
          \end{enumerate}
  \end{enumerate}  
\end{defi}

\begin{rem}
  \begin{enumerate}
    \item[]
    \item The above definition of topological spaces in the language of
          neighborhood bases is well known (see, e.g., \cite[Theorem 4.5]{Wil04}). 
    \item The above definition of pre-topological spaces can for example be
          found in \cite[Appendix A]{stadler02}. 
    \item Neighborhood frames satisfying \axCent\ are called centered
          \cite[Section 1.3]{Lew73}, hence the abbreviation.
  \end{enumerate}
\end{rem}

  In the language of topological spaces, neighborhood morphisms correspond to
  maps that are both continuous and open. Such maps are often called
  \emph{interior maps} \cite[Section~III.3]{RasSik63}. On the other hand,
  continuous maps are the ones that satisfy only the left-to-right implication
  of \eqref{eq:nbhd-mor}. 

  We write $\cat{PreTop_{int}}$ for the category of pre-topological spaces and
  interior maps, and $\cat{Top_{int}}$ for its full subcategory consisting of
  topological spaces. Clearly both $\cat{PreTop_{int}}$ and $\cat{Top_{int}}$
  are full subcategories of $\cat{FF}$. Moreover, we have:  
  $$
    \cat{PreTop_{int}} \cong \cat{FF}\axCent
      \quad\text{and}\quad
      \cat{Top_{int}} \cong \cat{FF}(\sf{Cent,iv}),
  $$
  where $\cat{FF}\axCent$ and $\cat{FF}(\sf{Cent,iv})$ are the full
  subcategories of $\cat{FF}$ whose objects satisfy $\axCent$ and
  $\sf{(Cent,iv)}$, respectively.

  The duals of topological spaces are given by complete atomic interior algebras,
  and the duals of pre-topological spaces by complete atomic pre-interior algebras.

\begin{defi}
  \begin{enumerate}
    \item[]
    \item A \emph{pre-interior algebra} is a modal algebra $(B, \dbox)$ that
          satisfies
          \begin{enumerate}
            \axitem{ax:T}{T} $\dbox b \leq b$.
          \end{enumerate}
          If $B$ is complete and atomic, $(B, \dbox)$ is a \emph{complete atomic
          pre-interior algebra}.
    \item An \emph{interior algebra} is a pre-interior algebra $(B, \dbox)$ that satisfies
          \begin{enumerate}
            \axitem{ax:four}{4} $\dbox\dbox b \leq \dbox b$.
          \end{enumerate}
          If $B$ is complete and atomic, $(B, \dbox)$ is a \emph{complete atomic
          interior algebra}.
  \end{enumerate}  
\end{defi}

\begin{rem}
  \begin{enumerate}
    \item[]
    \item The dual concept of interior operator is that of closure operator.
          Interior algebras were first introduced by McKinsey and Tarski
          \cite{McKTar44} in the language of closure operators and under the
          name of closure algebras. Rasiowa and Sikorski \cite{RasSik63} called
          these algebras topological Boolean algebras. The name interior algebra
          is due to Blok \cite{Blo76}.
    \item Generalizing closure on a powerset to pre-closure yields the notion
          of \emph{\v{C}ech closure spaces} \cite[Definition~14.A.1]{Cech66}.
          This provides an alternate language to talk about complete atomic
          pre-interior algebras.
  \end{enumerate}
\end{rem}

  Complete atomic pre-interior algebras are simply CANAs that satisfy
  \axN, \axC, and \axT.  We write $\cat{PreInt}$ for the full subcategory of
  $\cat{CANA}$ whose objects are pre-interior algebras. Then $\cat{PreInt}$
  is a full subcategory of $\cat{caMA}$. Let $\cat{Int}$ be the full
  subcategory of $\cat{PreInt}$ consisting of interior algebras.  

  Given a neighborhood frame $(X,N)$ a straightforward verification
  (see, e.g., \cite[Section 7.4]{Che80}) shows that:
  \begin{enumerate}
    \item[] $(X, N)$ validates \axCent\ iff its dual $(\wp X, \dbox_N)$ validates \axT.
    \item[] $(X, N)$ validates \axiv\ iff its dual $(\wp X, \dbox_N)$ validates \axfour.
  \end{enumerate}
  
  Thus, we arrive at the following duality theorems:
  
\begin{thm}\label{thm:dual-Pre-Top}
  The dual equivalence from Theorem \ref{thm:dual-FF-FIMA} restricts to
  $$
    \cat{PreInt} \equiv^{\op} \cat{PreTop_{int}}
      \quad\text{and}\quad
      \cat{Int} \equiv^{\op} \cat{Top_{int}}.
  $$
  \label{MT-for-open}
\end{thm}

\begin{rem}
  The object part of the dual equivalence $\cat{Int} \equiv^{\op} \cat{Top_{int}}$
  dates back to McKinsey and Tarski \cite{McKTar44}, and the morphism part
  to Rasiowa and Sikorski \cite[Section III.3]{RasSik63}. See \cite{BH21}
  for more details.
\end{rem}

\begin{rem}
  Restricting Theorem~\ref{MT-for-open} further gives rise to dualities for the categories
  of $T_0$-spaces, $T_1$-spaces, $P$-spaces (that is, topological spaces
  whose topology is closed under countable intersections),
  and Alexandrov spaces (topological spaces whose topology is closed
  under arbitrary intersections), with interior maps as morphisms.
\end{rem}

\begin{rem}
  Other topology-like spaces are the so-called
  ``generalized topological spaces'' of Cs\'{a}sz\'{a}r \cite{Csa02}.
  Proving correspondence results for the relevant axioms gives
  rise to a duality for such spaces in a similar manner as for
  pre-topological spaces.
\end{rem}

  The findings of this section can be summarized in the following diagram.
  The horizontal arrows indicate full inclusions of categories.
  The vertical arrows denote dual equivalences, and are labelled with the
  relevant theorem or example.

\bigskip
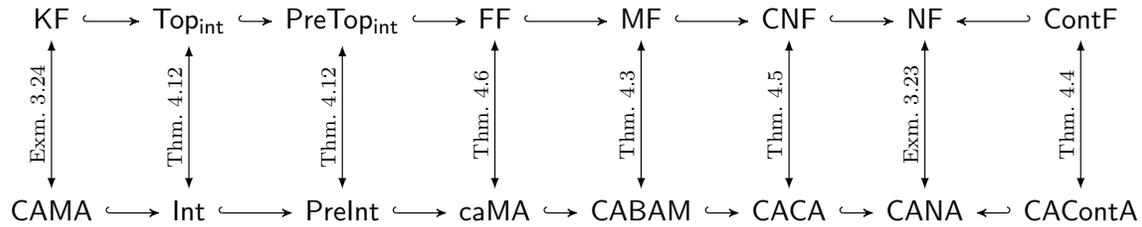
\begin{center}
  \begin{tikzcd}[column sep=1.2em, row sep=5em]
    \cat{KF} \arrow[r, hook]
        \arrow[d,latex-latex, "\text{Exm.~\ref{exm:fun-dual-Cinf}}" {rotate=90, anchor=south}]
      & \cat{Top_{int}}  \arrow[r, hook]
        \arrow[d,latex-latex, "\text{Thm.~\ref{thm:dual-Pre-Top}}" {rotate=90, anchor=south}]
      & \cat{PreTop_{int}}  \arrow[r, hook]
        \arrow[d,latex-latex, "\text{Thm.~\ref{thm:dual-Pre-Top}}" {rotate=90, anchor=south}]
      & \cat{FF}  \arrow[r, hook]
        \arrow[d,latex-latex, "\text{Thm.~\ref{thm:dual-FF-FIMA}}" {rotate=90, anchor=south}]
      & \cat{MF}  \arrow[r, hook]
        \arrow[d,latex-latex, "\text{Thm.~\ref{thm:dual-MF-CABAM}}" {rotate=90, anchor=south}]
      & \cat{CNF}  \arrow[r, hook]
        \arrow[d,latex-latex, "\text{Thm.~\ref{thm:dual-CNF-CACA}}" {rotate=90, anchor=south}]
      & \cat{NF}
        \arrow[d,latex-latex, "\text{Exm.~\ref{exm:fun-dual-no-axioms}}" {rotate=90, anchor=south}]
        \arrow[r, hookleftarrow]
      & \cat{ContF}
        \arrow[d,latex-latex, "\text{Thm.~\ref{thm:dual-cont}}" {rotate=90, anchor=south}] \\
    \cat{CAMA} \arrow[r, hook]
      & \cat{Int} \arrow[r, hook]
      & \cat{PreInt} \arrow[r, hook]
      & \cat{caMA} \arrow[r, hook]
      & \cat{CABAM} \arrow[r, hook]
      & \cat{CACA} \arrow[r, hook]
      & \cat{CANA} \arrow[r, hookleftarrow]
      & \cat{CAContA}
  \end{tikzcd}
\end{center}

%================================================================================
\section{J\'{o}nsson-Tarski type dualties}\label{sec:stone-type}

  In this section we derive categorical dualities for classes of neighborhood
  algebras that are not necessarily complete and atomic. While this simplifies
  the algebraic side of our story, it requires extra structure on the frame side
  of the duality: we now have to work with \emph{descriptive} neighborhood frames
  \cite{Dos89}. For this we work with one-step axioms in the standard modal
  language with finitary connectives. As corollaries we derive J\'{o}nsson-Tarski
  duality for modal algebras and Do\v{s}en duality for neighborhood algebras.
  
  Our main contribution is to define an analogue of $\fun{B}$ on Stone spaces by
  modifying the celebrated Vietoris construction. As a result, we obtain a new
  endofunctor on $\cat{Stone}$ and show that the category of coalgebras for this
  endofunctor is isomorphic to the category of descriptive neighborhood frames of 
  \cite[Section~2]{Dos89}. The Vietoris space of a Stone space $\topo{X}$ is
  embeddable in this new hyperspace of $\topo{X}$ as a closed subspace.

%--------------------------------------------------------------------------------
\subsection{Descriptive neighborhood frames}\label{subsec:dnf-coalg}

  We start by recalling the definition of a descriptive
  neighborhood frame, first introduced by Do\v{s}en \cite[Section 2]{Dos89}.
  However, to be in line with standard practice in modal logic, our definition of \emph{general}
  frames deviates slightly from that of Do\v{s}en in that
  Do\v{s}en's definition requires that all neighborhoods
  are admissible, while we view this as an additional \emph{tightness} condition.
  So our \emph{tight general frames} correspond to Do\v{s}en's \emph{general frames}.
  
  Subsequently, we define a new endofunctor on $\cat{Stone}$ generalizing the
  Vietoris endofunctor, and show that the category of descriptive neighborhood
  frames can be viewed as the category of coalgebras for this endofunctor. 
  
\begin{defi}\label{def:gen}
  \begin{enumerate}
    \item[]
    \item \label{it:gen-frm} 
          A \emph{general \textup{(}neighborhood\textup{)} frame} is a tuple $(X, N, A)$
          consisting of a neighborhood frame $(X, N)$ and a Boolean subalgebra
          $A$ of $\fun{P}X$ such that $A$ is closed under
          $\dbox_N : \fun{P}X \to \fun{P}X$ given by
          $$
            \dbox_N a = \{ x \in X \mid a \in N(x) \}.
          $$
  \item A general frame $(X,N,A)$ is \emph{tight} if $N(x) \subseteq A$
  for all $x\in X$.
  \item A general frame $(X, N, A)$ is 
  \emph{differentiated} if for all distinct $x, y \in X$ there is
          $a \in A$ such that $x \in a$ and $y \notin a$; and
  \emph{compact} if whenever $A' \subseteq A$ has the finite intersection
          property, then $\bigcap A' \neq \varnothing$.
  \item A \emph{descriptive \textup{(}neighborhood\textup{)} frame} is a general frame that is
  differentiated, compact, and tight. 
  \item A \emph{general frame morphism} from $(X, N, A)$ to $(X', N', A')$ is a neighborhood
  frame morphism $f : (X, N) \to (X', N')$ such that $f^{-1}(a') \in A$ for all
  $a' \in A'$. We denote the category of descriptive frames
  and general frame morphisms by $\cat{DNF}$.
  \end{enumerate}
\end{defi}
  
  Let $(X, N, A)$ be a descriptive frame. As usual, we can generate a topology
  $\tau_A$ on $X$ using $A$ as a base. The space $(X, \tau_A)$ is compact
  and Hausdorff because $(X, N, A)$ is compact and differentiated.
  Moreover, it is zero-dimensional because $A$ is closed under complements,
  and hence $(X, \tau_A)$ is a Stone space.
  A subset $a \subseteq X$ is in $A$ iff it is
  clopen in $(X, \tau_A)$, so we can recover $A$ from the Stone topology.

  Since $A = \fun{clp}\topo{X}$ and $N(x) \subseteq A$, it makes sense to
  define a functor $\fun{D}$ on $\cat{Stone}$ that sends $\topo{X} \in \cat{Stone}$
  to $\fun{P}(\fun{clp}\topo{X})$. The choice of topology on
  $\fun{P}(\fun{clp}\topo{X})$ is motivated by the desideratum to turn
  $\fun{D}\topo{X}$ into a Stone space, and is a generalization of the
  Vietoris topology.

\begin{defi}
  The \emph{$\fun{D}$-hyperspace} $\fun{D}\topo{X}$ of a Stone space $\topo{X}$
  is the space $\fun{P}(\fun{clp}\topo{X})$ whose topology is generated by
  the clopen subbase
  $$
    \lbox a = \{ W \in \fun{P}(\fun{clp}\topo{X}) \mid a \in W \},
    \qquad
    \ldiamond a = \{ W \in \fun{P}(\fun{clp}\topo{X}) \mid X \setminus a \notin W \},
  $$
  where $a$ ranges over the clopen subsets of $\topo{X}$.
  For a continuous function $f : \topo{X} \to \topo{X}'$ between Stone spaces,
  define $\fun{D}f : \fun{D}\topo{X} \to \fun{D}\topo{X}'$ by
  $$
    \fun{D}f(W) = \{ a' \in \fun{clp}\topo{X}' \mid f^{-1}(a') \in W \}.
  $$
\end{defi}

\begin{lem}
  The assignment $\fun{D}$ defines an endofunctor on $\cat{Stone}$.
\end{lem}

\begin{proof}
  To see that $\fun{D}$ is well defined, we first show that $\fun{D}\topo{X}$
  is a Stone space. 
  Zero-dimensionality of $\fun{D}\topo{X}$ follows from the fact that it
  is generated by a base that is closed under complementation.
  (Indeed, for all $a \in \fun{clp}\topo{X}$ we have
  $\fun{D}\topo{X} \setminus \lbox a = \ldiamond(X \setminus a)$
  and $\fun{D}\topo{X} \setminus \ldiamond a = \lbox(X \setminus a)$.)
  To see that $\fun{D}\topo{X}$ is Hausdorff, suppose that $W, W' \in \fun{D}\topo{X}$ are
  distinct. Then there must be an $a \in \fun{clp}\topo{X}$ such
  that either $a \in W$ and $a \notin W'$, or $a \notin W$ and $a \in W'$.
  In either case $\lbox a$ and $\ldiamond(X \setminus a)$ provide two disjoint
  open subsets of $\fun{D}\topo{X}$ separating $W$ and $W'$.
  
  For compactness, by the Alexander subbase theorem, it suffices to prove that every open cover
  consisting of subbasic (cl)opens has a finite subcover. So suppose
  \begin{equation}\label{eq:cover}
    \fun{D}\topo{X} = \bigcup_{a \in A} \lbox a \cup \bigcup_{b \in B} \ldiamond b,
  \end{equation}
  where $A, B \subseteq \fun{clp}\topo{X}$.
  Consider
  $$
    W = \{ X \setminus b \mid b \in B \} \in \fun{D}\topo{X}.
  $$
  By construction, this is in none of the $\ldiamond b$, so there must be
  $a' \in A$ such that $W \in \lbox a'$. But this means $a' = X \setminus b'$ for
  some $b' \in B$.
  Consequently, if $V$ is an arbitrary element of $\fun{D}\topo{X}$ such that
  $V \notin \lbox a'$, then
  $V \in \fun{D}\topo{X} \setminus \lbox a' = \ldiamond(X \setminus a') = \ldiamond b'$.
  Therefore,
  $$
    \fun{D}\topo{X} = \lbox a' \cup \ldiamond b',
  $$
  so we have found a finite subcover of the cover in \eqref{eq:cover}.
  Thus, $\fun{D}\topo{X}$ is a Stone space.
  
  Finally, we show that $\fun{D}$ is well defined on morphisms.
  Let $f : \topo{X} \to \topo{X}'$ be a morphism in $\cat{Stone}$.
  In order to prove that $\fun{D}f$ is continuous, it suffices to show that
  $(\fun{D}f)^{-1}(\lbox a')$ is clopen in $\fun{D}\topo{X}$ for all $a' \in \fun{clp}\topo{X}'$.
  (The case for diamonds follows by working with complements.)
  So let $a' \in \fun{clp}\topo{X}'$. Then
  $$
    (\fun{D}f)^{-1}(\lbox a')
      = \{ W \in \fun{D}\topo{X} \mid a' \in \fun{D}f(W) \}
      = \{ W \in \fun{D}\topo{X} \mid f^{-1}(a') \in W \}
      = \lbox f^{-1}(a'),
  $$
  which is clopen in $\fun{D}\topo{X}$. Consequently, $\fun{D}$ is well defined. 
  Functoriality follows from the fact that $\fun{UD}$ is a subfunctor of $\fun{BU}$,
  where $\fun{B} = \contra\contra$ and $\fun{U} : \cat{Stone} \to \cat{Set}$ is the forgetful functor.
\end{proof}

\begin{thm}\label{thm:dnf-coalg}
  $
    \cat{DNF} \cong \cat{Coalg}(\fun{D}).
  $
\end{thm}

\begin{proof}
  For isomorphism on objects,
  if $(X, N, A)$ is a descriptive neighborhood frame and $\tau_A$ is the 
  topology on $X$ generated by $A$, then $N$ is a function from
  $\topo{X} = (X, \tau_A)$ to $\fun{D}\topo{X}$ which is continuous
  because $A$ is closed under $\dbox_N$ defined in Definition~\ref{def:gen}\eqref{it:gen-frm}.
  Thus, $(\topo{X}, N)$ is a $\fun{D}$-coalgebra.
  
  Conversely, a $\fun{D}$-coalgebra $(\topo{X}, \gamma)$ gives rise to the
  descriptive neighborhood frame $(X, N, \fun{clp}\topo{X})$, where
  $X$ is the set underlying $\topo{X}$ and $N$ is defined by
  $N(x) = \gamma(x)$. 
  Continuity of $\gamma$ entails that $\fun{clp}\topo{X}$ is closed under $\dbox_N$.
  
  The isomorphism on morphisms follows from a straightforward verification.
\end{proof}

\begin{rem}\label{rem:V-vs-D}
  The Vietoris functor $\fun{V}$ is a subfunctor of $\fun{D}$ via the
  natural transformation $\zeta : \fun{V} \to \fun{D}$ defined on components by
  $$
    \zeta_{\topo{X}}
      : \fun{V}\topo{X} \to \fun{D}\topo{X}
      : c \mapsto \{ a \in \fun{clp}\topo{X} \mid c \subseteq a \}.
  $$
  This gives rise to a functor $\hat{\zeta} : \cat{Coalg}(\fun{V}) \to \cat{Coalg}(\fun{D})$,
  defined on objects by sending $(\topo{X}, \gamma) \in \cat{Coalg}(\fun{V})$ to
  $$
    \begin{tikzcd}
      \topo{X}
            \arrow[r, "\gamma"]
        & \fun{V}\topo{X}
            \arrow[r, "\zeta_{\topo{X}}"]
        & \fun{D}\topo{X}
    \end{tikzcd}
  $$
  and on morphisms by $\hat{\zeta}f = f$.
  Specifically, if $(\topo{X}, \gamma)$ is a $\fun{V}$-coalgebra, then the
  corresponding $\fun{D}$-coalgebra is given by $(\topo{X}, N)$, where
  $N(x) = \{ a \in \fun{clp}\topo{X} \mid \gamma(x) \subseteq a \}$.
  The descriptive neighborhood frames lying in the image of $\hat{\zeta}$ are
  precisely those descriptive frames validating \axN\ and \axC;
  see also Example \ref{exm:JT}.
\end{rem}

%--------------------------------------------------------------------------------
\subsection{Functor dualities}

  Like we did in Section \ref{subsec:fdt-set}, we can use one-step axioms to
  prove that certain quotient functors of $\fun{N}$ from 
  Proposition~\ref{prop:alg}\eqref{it:prop:alg-N} are Stone-dual to subfunctors
  of~$\fun{D}$.
  
  Since $\fun{N}$ is an endofunctor on $\cat{BA}$, we can only take quotients
  with finitary axioms. These are axioms that have finite conjunctions
  and disjunctions. Since this implies that axioms can only contain
  a finite number of variables, it suffices to work with a countable set of
  variables. Thus, we will work with the standard modal language, viewed
  as a sublanguage of $\mf{L}_{\Box}$.

\begin{defi}
  A \emph{finitary axiom} is a formula in the language $\mf{L}_{\Box}^{\omega}$
  (defined as in Definition \ref{def:inf-lan}).
  A \emph{finitary one-step axiom} is a formula in the language
  $(\mf{L}^{\omega}_{\Box})^1$.
\end{defi}

  We define assignments and satisfaction of these axioms in neighborhood algebras
  as in Definition \ref{def:ass-sat}. That is, an assignment for a neighborhood
  algebra $(B, \dbox)$ is a function $\theta : V_{\omega} \to B$ and extends uniquely to
  a map $\hat{\theta} : \mf{L}_{\Box}^{\omega} \to B$.
  We write $(B, \dbox) \Vdash \phi$ if $\hat{\theta}(\phi) = 1$ for
  every assignment $\theta$.
  
  Furthermore, the assignment $\theta$ gives rise to a map
  $\ov{\theta} : (\mf{L}_{\Box}^{\omega})^1 \to \fun{N}B$ which interprets
  finitary one-step axioms in $\fun{N}B$ in the same manner as in Definition
  \ref{def:theta-ov}.

\begin{defi}
  Let $\Ax$ be a collection of finitary one-step axioms.
  For $B \in \cat{BA}$, define $\fun{N}_{\Ax}B$ to be the free Boolean algebra
  generated by $\{ \Box b \mid b \in B \}$ modulo the congruence relation $\sim_{\Ax}$
  generated by $\{ \ov{\theta}(\phi) \sim_{\Ax} 1 \}$,
  where $\phi$ ranges over $\Ax$ and $\theta$ over the assignments
  $V_{\omega} \to B$.
  For a homomorphism $h : B \to B'$ define $\fun{N}_{\Ax}h$ on generators
  by $\fun{N}_{\Ax}h(\Box b) = \Box h(b)$.
  Then $\fun{N}_{\Ax}$ defines a functor $\cat{BA} \to \cat{BA}$.
\end{defi}

\begin{exa}
  Again, well-known functors can be obtained via this procedure.
  Of course, if we take $\Ax = \varnothing$, then we get
  $\fun{N}_{\emptyset} = \fun{N}$.
  Similarly, the axioms
  \begin{enumerate}
    \item [\axN]
          $\Box\top$
    \axitem{ax:C}{C}
          $\Box v \wedge \Box v' \leftrightarrow \Box(v \wedge v')$
  \end{enumerate}
  give rise to the endofunctor on $\cat{BA}$ whose algebras are normal modal
  algebras.
\end{exa}

  Like in Section \ref{sec:thomason}, for each $\fun{N}_{\Ax}$ we can
  define a dual functor $\fun{D}_{\Ax}$.
  The functor $\fun{D}_{\Ax}$ arises as a subfunctor of $\fun{D}$.
  In particular, this means that for every Stone space $\topo{X}$,
  the space $\fun{D}_{\Ax}\topo{X}$ is a subspace of $\fun{D}\topo{X}$.

\begin{defi}\label{def:ax-subset-stone}
  Let $\topo{X}$ be a Stone space and $B$ its dual Boolean algebra of clopens.
  For a finitary one-step axiom $\phi$ and an assignment $\theta : V_{\omega} \to B$
  of the variables we define the clopen neighborhood $\theta^t(\phi)$
  recursively by
  \begin{align*}
    \theta^t(\Box v) &= \lbox \theta(v) \\
    \theta^t(\top) &= \fun{P}B \\
    \theta^t(\neg\phi) &= \fun{P}B \setminus \theta^t(\phi) \\
    \theta^t(\phi \wedge \psi) &= \theta^t(\phi) \cap \theta^t(\psi)
  \end{align*}
  Let $\phi$ be a finitary one-step axiom.
  Then we call $W \in \fun{D}\topo{X}$ a $\phi$-subset of $B$ if $W \in \theta^t(\phi)$
  for every assignment $\theta$ of the variables $V_{\omega}$.
  If $\Ax$ is a collection of finitary one-step axioms, then we say that
  $W$ is an \emph{$\Ax$-subset} if $W$ is a $\phi$-subset
  for all axioms $\phi \in \Ax$.
\end{defi}

  The next lemma is an analogue of Lemma \ref{lem:uf-set}.
  %in the setting of descriptive frames.

\begin{lem}\label{lem:uf-stone}
  Ultrafilters of $\fun{N}_{\Ax}B$ correspond bijectively to $\Ax$-subsets of $B$.
\end{lem}
\begin{proof}
  The bijection is established by sending an ultrafilter of $\fun{N}_{\Ax}B$
  (viewed as a homomorphism $p : \fun{N}_{\Ax}B \to 2$) to the set
  \begin{equation}\label{eq:uf-to-set-descr}
    W_p = \{ b \in B \mid p(\Box b) = 1 \} \subseteq B,
  \end{equation}
  and by sending any $\Ax$-subset $W \subseteq B$ to the map
  $p_W : \fun{N}_{\Ax}B \to 2$ defined on generators by $p_W(\Box b) = 1$ iff $b \in W$.
  The remainder of the proof can be obtained from the proof of Lemma \ref{lem:uf-set} by
  replacing ``$\fun{L}$'' with $``\fun{N}$'',
  ``$\fun{L}_{\Ax}$'' with ``$\fun{N}_{\Ax}$'',
  ``complete homomorphism'' with ``homomorphism'',
  and ``complete congruence'' with ``congruence''.
\end{proof}

  We can now define the functor $\fun{D}_{\Ax}$ in a similar manner to
  $\fun{B}_{\Ax}$ from Definition \ref{def:fun-B}.

\begin{defi}
  For a Stone space $\topo{X}$, let $\fun{D}_{\Ax}\topo{X}$ be the subspace
  of $\fun{D}\topo{X}$ whose objects are $\Ax$-subsets of $\fun{clp}\topo{X}$.
  For a continuous function $f : \topo{X} \to \topo{X}'$ let
  $\fun{D}_{\Ax}f$ be the restriction of $\fun{D}f$ to $\fun{D}_{\Ax}\topo{X}$.
\end{defi}

  While for any set $X$ the set $\fun{B}_{\Ax}X$ is automatically a subset
  of $\fun{B}X$, in our current setting it is not guaranteed that
  $\fun{D}_{\Ax}\topo{X}$ is a Stone subspace of $\fun{D}\topo{X}$.
  The next proposition ensures that this is indeed the case.

\begin{prop}
  $\fun{D}_{\Ax}$ is an endofunctor on $\cat{Stone}$.
\end{prop}
\begin{proof}
  Let $\topo{X}$ be a Stone space. In order to show that $\fun{D}_{\Ax}\topo{X}$
  is again a Stone space it suffices to prove that it is a closed subspace
  of $\fun{D}\topo{X}$.
  Note that $\theta^t(\phi)$ is a clopen subset of $\fun{D}\topo{X}$ for any
  finitary one-step axiom $\phi$ and assignment $\theta : V_{\omega} \to \fun{clp}\topo{X}$.
  By definition, the set underlying $\fun{D}_{\Ax}\topo{X}$ is given by
  \begin{equation}\label{eq:uf-stone}
    \bigcap \{ \theta^t(\phi) \mid \phi \in \Ax, \theta : V_{\omega} \to \fun{clp}\topo{X} \} \subseteq \fun{D}\topo{X}.
  \end{equation}
  Since this is the intersection of clopen subsets of $\fun{D}\topo{X}$,
  it is a closed subset of $\fun{D}\topo{X}$.
  
  That $\fun{D}_{\Ax}$ is well defined on morphisms can be proven as in
  Proposition~\ref{prop:Bax-welldef}, and functoriality of $\fun{D}_{\Ax}$
  follows from functoriality of $\fun{D}$.
\end{proof}

  The next theorem is an analogue of Theorem \ref{thm:fun-dual} in the setting
  of descriptive frames.

\begin{thm}[Stone Functor Duality Theorem]\label{thm:fun-dual-stone}
  Let $\Ax$ be a collection of finitary one-step axioms.
  Then the functors $\fun{N}_{\Ax}$ and $\fun{D}_{\Ax}$ are Stone-duals.
\end{thm}

\begin{proof}
  We view ultrafilters as homomorphisms, so for example when we say
  $p \in \fun{uf}(\fun{N}_{\Ax}(\fun{clp}\topo{X}))$ we view
  $p$ as a homomorphism $\fun{N}_{\Ax}(\fun{clp}\topo{X}) \to 2$.
  For a Stone space $\topo{X}$, define the function
  $\xi_{\topo{X}} : \fun{uf}(\fun{N}_{\Ax}(\fun{clp}\topo{X})) \to \fun{D}_{\Ax}\topo{X}$ 
  by $p \mapsto W_p$ (defined as in \eqref{eq:uf-to-set-descr}).
  This yields a bijection on objects by Lemma \ref{lem:uf-stone}.
  
  To see that $\xi_{\topo{X}}$ is continuous,
  recall that the topology on $\fun{D}\topo{X}$ is generated by
  $\lbox a$ and $\ldiamond a$, where $a$ ranges over the clopens of $\topo{X}$.
  So the topology on $\fun{D}_{\Ax}\topo{X}$ is generated by
  $$
    \lbox a \cap \fun{D}_{\Ax}\topo{X}, \qquad \ldiamond a \cap \fun{D}_{\Ax}\topo{X},
  $$
  where $a$ ranges over the clopen subsets of $\topo{X}$.
  Also, recall that for any Boolean algebra $B$, the topology on
  $\fun{uf}B$ is generated by sets of the form
  $\beta(a) = \{ p \in \fun{uf}B \mid p(a) = 1 \}$, where $a \in B$.
  Continuity of $\xi_{\topo{X}}$ now follows from the fact that
  $$
    \xi_{\topo{X}}^{-1}(\lbox a \cap \fun{D}_{\Ax}(\topo{X}))
      = \{ p \in \fun{uf}(\fun{N}_{\Ax}(\fun{clp}\topo{X})) \mid p(\Box a) = 1 \}
      = \beta(\Box a)
  $$
  is open in $\fun{uf}(\fun{N}_{\Ax}(\fun{clp}\topo{X}))$,
  and similarly $\xi_{\topo{X}}^{-1}(\ldiamond a) = \beta(\neg\Box \neg a)$.
  Since $\xi_{\topo{X}}$ is a continuous bijection between Stone spaces, it
  is a homeomorphism, i.e., an isomorphism in $\cat{Stone}$.
  
  Finally, we prove that the assignment
  $\xi = (\xi_{\topo{X}})_{\topo{X} \in \cat{Stone}}
    : \fun{uf} \circ \fun{N}_{\Ax} \circ \fun{clp} \to \fun{D}_{\Ax}$
  is natural by showing that
  for every continuous function $f : \topo{X} \to \topo{X}'$ the diagram
  $$
    \begin{tikzcd}
      \fun{uf}(\fun{N}_{\Ax}(\fun{clp}\topo{X})) 
            \arrow[d, "\fun{uf}(\fun{N}_{\Ax}(\fun{clp}f))" swap]
            \arrow[r, "\xi_{\topo{X}}"]
        & [1em]
          \fun{D}_{\Ax}\topo{X}
            \arrow[d, "\fun{D}_{\Ax}f"] \\ [1em]
      \fun{uf}(\fun{N}_{\Ax}(\fun{clp}\topo{X}'))
            \arrow[r, "\xi_{\topo{X}'}"]
        & \fun{D}_{\Ax}\topo{X}'
    \end{tikzcd}
  $$
  commutes.
  To this end, let $p \in \fun{uf}(\fun{N}_{\Ax}(\fun{clp}\topo{X}))$
  and $a' \in \fun{clp}\topo{X}'$.
  Then
  \begin{align*}
    a' \in \fun{D}_{\Ax}f \circ \xi_{\topo{X}}(p)
      &\iff \xi_{\topo{X}}(p)(f^{-1}(a')) = 1 \\
      &\iff p(\Box f^{-1}(a')) = 1\\
      &\iff p(\fun{N}_{\Ax}(\fun{clp}f)(\Box a')) = 1 \\
      &\iff \fun{uf}(\fun{N}_{\Ax}(\fun{clp}f))(p)(\Box a') = 1 \\
      &\iff a' \in \xi_{\topo{X}'} \circ \fun{uf}(\fun{N}_{\Ax}(\fun{clp}f))(p).
  \end{align*}
  This proves the theorem.
\end{proof}

\begin{cor}\label{cor:fun-dual-stone}
  For every set $\Ax$ of finitary one-step axioms, we have
  $$
    \cat{Alg}(\fun{N}_{\Ax}) \equiv^{\op} \cat{Coalg}(\fun{D}_{\Ax}).
  $$
\end{cor}

  Some well-known dualities are instantiations of Corollary~\ref{cor:fun-dual-stone}.

\begin{exa}
  If we take $\Ax = \varnothing$, then we recover Do\v{s}en's duality for descriptive
  neighborhood frames \cite[Theorem~6]{Dos89}:
  $$
    \cat{NA}
      \cong \cat{Alg}(\fun{N})
      \equiv^{\op} \cat{Coalg}(\fun{D})
      \cong \cat{DNF}.
  $$
\end{exa}

\begin{exa}\label{exm:JT}
  We can derive J\'{o}nsson-Tarski duality for normal modal algebras
  as follows. As we pointed out in Remark \ref{rem:V-vs-D}, 
  $\fun{V} = \fun{D}_{\{ \axN, \axC \}}$.
  Moreover, it is immediate from the definition of $\fun{N}_{\{ \axN, \axC\}}$
  that $\cat{MA} \cong \cat{Alg}(\fun{N}_{\{ \axN, \axC \}})$.
  Therefore,
  $$
    \cat{MA}
      \cong \cat{Alg}(\fun{N}_{\{ \axN, \axC \}})
      \equiv^{\op} \cat{Coalg}(\fun{D}_{\{ \axN, \axC \}})
      \cong \cat{Coalg}(\fun{V})
      \cong \cat{DKF}.
  $$
\end{exa}

  Corollary \ref{cor:fun-dual-stone} also gives rise to the notion of
  a descriptive contingency frame, as shown in the next example.

\begin{exa}\label{exm:dual-descr-cont}
  Since \axCont\ is a finitary axiom, it gives rise to the notions of contingency neighborhood algebras,
  descriptive contingency neighborhood frames, and a duality between them.
  A \emph{contingency neighborhood algebra} is a neighborhood
  algebra $(B, \dbox)$ such that $\dbox b = \dbox \neg b$. We write
  $\cat{CNA}$ for the full subcategory of $\cat{NA}$ of contingency neighborhood algebras.
  A \emph{descriptive contingency neighborhood frame} is a $\fun{D}_{\axCont}$-coalgebra.
  Explicitly we can describe these as tuples $(X, N, A)$ such that
  $(X, N, A)$ is a descriptive neighborhood frame and
  $(X, N)$ is a contingency neighborhood frame.
  Writing $\cat{DCNF}$ for the full subcategory of $\cat{DNF}$
  whose objects are descriptive contingency neighborhood frames,
  we obtain:
  $$
    \cat{CNA}
      \cong \cat{Alg}(\fun{N}_{\axCont})
      \equiv^{\op} \cat{Coalg}(\fun{D}_{\axCont})
      \cong \cat{DCNF}.
  $$
\end{exa}

%--------------------------------------------------------------------------------
\subsection{Forgetting incorrectly}\label{subsec:problem}

  Just like \axCont\ in Example \ref{exm:dual-descr-cont},
  the axiom \axM\ is also finitary. Therefore, we get functors
  $\fun{N}_{\axM}$ and $\fun{D}_{\axM}$ (shorthand for $\fun{N}_{\{ \axM \}}$ and
  $\fun{D}_{\{ \axM \}}$) such that
  $\cat{Alg}(\fun{N}_{\axM}) \equiv^{\op} \cat{Coalg}(\fun{D}_{\axM})$.
  Moreover, we have
  $$
    \cat{Alg}(\fun{N}_{\axM}) \cong \cat{BAM}.
  $$
  Thus, one may think that the coalgebras in $\cat{Coalg}(\fun{D}_{\axM})$
  give a suitable notion of descriptive monotone frames.
  However, care is needed:
  if we take a $\fun{D}_{\axM}$-coalgebra and
  forget about the topology, we do \emph{not} obtain a $\fun{B}_{\axM}$-coalgebra.
  Indeed, since only clopen sets are allowed to serve as neighborhoods
  of a state, the collection of neighborhoods at a state need not be
  upward closed under inclusion.
  We give an example of this phenomenon.

\begin{exa}\label{exm:mon-non-mon}
  Let $X = \mb{N} \cup \{ \infty \}$ and generate a topology $\tau$ on $X$ by
  the finite subsets of $\mb{N}$ and cofinite sets containing $\infty$. Thus, $\topo{X} = (X, \tau)$ is the one-point compactification of the discrete space $\mb{N}$.
  Clearly $\topo{X}$ is a Stone space.
  Define $\gamma : \topo{X} \to \fun{D}_{\axM}\topo{X}$ by
  $$
    \gamma(x) = \fun{clp}\topo{X}.
  $$
  Then $(\topo{X}, \gamma)$ is a $\fun{D}_{\axM}$-coalgebra.
  However, $(X, \gamma)$ does not define a $\fun{B}_{\axM}$-coalgebra
  because $\mb{N}_{\odd} = \{ x \in \mb{N} \mid x \text{ is odd} \}$
  is not clopen, hence it is not in $\gamma(x)$, while both $\varnothing \in \gamma(x)$
  and $\varnothing \subseteq \mb{N}_{\odd}$.
\end{exa}

  A similar problem occurs with descriptive convex frames.
  In fact, Example \ref{exm:mon-non-mon} also shows that 
  if we forget about the topological structure of a $\fun{D}_{\axConv}$-coalgebra,
  then we do not necessarily end up with a $\fun{B}_{\axConv}$-coalgebra
  because $\varnothing \subseteq \mb{N}_{\odd} \subseteq X$
  and $\varnothing, X \in \gamma(x)$.
  One way to remedy this is by looking at canonical extensions.
  Such an approach was carried out in \cite{Han03} for the special case of
  monotone frames.
  We explore this idea in the next section.

%================================================================================
\section{Canonical extensions}\label{sec:canonical}

  In this section we discuss the $\sigma$- and $\pi$-extensions of (descriptive)
  neighborhood algebras. In general, these extensions are not functorial.  
  However, we show that adding the convexity and co-convexity axioms
  yields functoriality of the $\sigma$- and $\pi$-extensions, respectively.
  We also show that adding these axioms allows us to view descriptive
  frames as categories of coalgebras.
  We use this to give an alternative coalgebraic proof of the duality for
  monotone Boolean algebra expansions of \cite{HanKup04}. 
 
%--------------------------------------------------------------------------------
\subsection{$\sigma$- and $\pi$-exensions}

  There is an obvious forgetful functor $\cat{DNF} \to \cat{NF}$.
  Although this forgetful functor is conveniently simple, it has an undesirable
  property: it ``leaves gaps.'' For example, if $(X, N, A)$ is a monotone
  descriptive frame, then $(X, N)$ is \emph{not} generally a monotone frame.
  Indeed, monotonicity now only holds with respect to admissible sets (clopens).
  Similarly, if $(X, N, A)$ is a normal descriptive frame and $R_N$ is the
  relation that arises from $N$ via $R_N(x) = \bigcap N(x)$, then we would prefer
  that the underlying neighborhood frame be normal as well.
  This would be the case if $(X, N)$ satisfied $a \in N(x)$ iff
  $R_N[x] \subseteq a$ for each $a \subseteq X$.
  Again, since $N(x)$ contains only admissible subsets
  of $X$, this need not be the case.%
    \footnote{For example, let $\topo{X} = \mb{N} \cup \{ \infty \}$ be the
    one-point compactification of the discrete space $\mb{N}$ and $X$
    the set underlying $\topo{X}$. For each $x \in X$ let $N(x)$ be the
    collection of co-finite subsets of $X$ containing $\infty$. Then it is easy
    to see that $(X, N, \fun{clp}\topo{X})$ is a descriptive neighborhood frame,
    that $R_N(x) = \{ \infty \}$, but that $\{ \infty \} \notin N(x)$.}
  
  To remedy this, we explore alternatives to the forgetful functor.
  These are dual versions of the well-known $\sigma$- and $\pi$-extensions
  from the theory of canonical extensions
  (see, e.g.,~\cite{JonTar51,GehJon94,GehHar01,GehJon04}):

\begin{defi}
  Let $A$ be a CABA and $B$ a Boolean subalgebra of $A$.
  \begin{enumerate}
    \item $B$ is called \emph{dense} in $A$ if every element in $A$ is the 
          join of meets of elements in $B$.
    \item We say that $B$ is \emph{compact} in $A$ if for all sets
          $S, T \subseteq B$ with $\bigwedge S \leq \bigvee T$ in $A$,
          there exist finite $S' \subseteq S$ and $T' \subseteq T$ such that
          $\bigwedge S' \leq \bigvee T'$.
    \item  A \emph{canonical extension} of a Boolean algebra $B$ is a pair
          $(B^\sigma, e)$ where $B^\sigma$ is a CABA and
          $e : B \to B^\sigma$ is a Boolean embedding such that $e[B]$ is
          dense and compact in $B^\sigma$.
  \end{enumerate}
\end{defi}
  
\begin{rem}
  It is well known \cite{JonTar51}
  that the canonical extension of a Boolean algebra is unique up to isomorphism,
  and can explicitly be described as the powerset of the dual Stone space of $B$.
  Thus, we often speak of $B^{\sigma}$ as \emph{the} canonical extension of $B$
  and view $B$ as sitting inside $B^{\sigma}$.
\end{rem}

  Let $\dbox : B \to B$ be an endofunction (not necessarily a homomorphism).
  Then we can extend $\dbox$ to a function
  $B^{\sigma} \to B^{\sigma}$ in several ways.
  Two well-known extensions are the $\sigma$- and $\pi$-extensions.
  To define these, we recall the notions of closed and open elements of
  $B^{\sigma}$ \cite[Definition 1.20]{JonTar51}.
  
  We say that $x \in B^{\sigma}$ is \emph{closed} if it is a meet
  of elements from $B$, and $x$ is \emph{open} if it is a join of elements from $B$.
  We write $\fun{K}_B$ and $\fun{O}_B$ for the sets of closed and
  open elements of $B^{\sigma}$, respectively.
  The $\sigma$- and $\pi$-extensions of $\dbox : B \to B$ are
  now defined by
  \begin{align*}
    \dbox^{\sigma} x
      &= \bigvee \big\{ \bigwedge \{ \dbox b \mid b \in B \text{ and } c \leq b \leq d \}
        \mid c \in \fun{K}_B, d \in \fun{O}_B, c \leq x \leq d \big\} \\
    \dbox^{\pi} x
      &= \bigwedge \big\{ \bigvee \{ \dbox b \mid b \in B \text{ and } c \leq b \leq d \}
        \mid c \in \fun{K}_B, d \in \fun{O}_B, c \leq x \leq d \big\}
  \end{align*}
  These give maps $\sigma, \pi : \cat{NA} \to \cat{CANA}$, which in turn give rise to
  $\Sigma, \Pi : \cat{DNF} \to \cat{NF}$ by composing as follows:
  \[
    \Sigma = \fun{at} \circ \sigma \circ \fun{clp} \quad\text{and}\quad
    \Pi = \fun{at} \circ \pi \circ \fun{clp}.
  \]
  Thus we have the following diagram:
  \begin{equation}\label{eq:sigma-pi}
    \begin{tikzcd}
      \cat{NA}
            \arrow[r, shift left=2pt, "\fun{uf}"]
            \arrow[d, shift right=2pt, bend right=10, "\sigma" swap]
            \arrow[d, shift left=2pt, bend left=10, "\pi"]
        & [1em]
          \cat{DNF}
            \arrow[l, shift left=2pt, "\fun{clp}"]
            \arrow[d, dashed, shift right=2pt, bend right=10, "\Sigma" swap]
            \arrow[d, dashed, shift left=2pt, bend left=10, "\Pi"] \\ [1em]
      \cat{CANA}
            \arrow[r, shift left=2pt, "\fun{at}"]
        & \cat{NF}
            \arrow[l, shift left=2pt, "\wp"]
    \end{tikzcd}
  \end{equation}

  We point out that $\sigma, \pi : \cat{NA} \to \cat{CANA}$
  are not necessarily functors (see \cite[Example 3.4 and Remark 3.5]{BMM08}),
  and hence neither are $\Sigma$ and $\Pi$.
  We will temporarily ignore this issue and focus solely on the action of
  $\sigma$ and $\pi$ on objects.

  We next define $\Sigma$ and $\Pi$ explicitly.
  For this we need the notions of closed and open elements of $(X, N, A)$.
  These are defined to be the closed and open sets of the topological space
  $\topo{X}=(X,\tau_A)$, and denoted by $\fun{K}_A$ and $\fun{O}_A$, respectively.
  Finally, for $c, d \in \fun{P}X$ define
  $[c, d] = \{ e \in \fun{P}X \mid c \subseteq e \subseteq d \}$.
  
\begin{defi}
  Let $(X, N, A)$ be a descriptive neighborhood frame.
  \begin{enumerate}
\item Define the $\sigma$-extension of $N$ by
  $$
    N^{\sigma}(x)
      = \{ e \in \fun{P}X \mid \exists c \in \fun{K}_A, d \in \fun{O}_A
        \text{ with } c \subseteq e \subseteq d
        \text{ and } [c, d] \cap A \subseteq N(x) \},
  $$
  and set $\Sigma(X, N, A) = (X, N^{\sigma})$.
  
\item  Define the $\pi$-extension of $N$ by
  $$
    N^{\pi}(x)
      = \{ e \in \fun{P}X \mid \forall c \in \fun{K}_A, d \in \fun{O}_A
        \text{ with } c \subseteq e \subseteq d
        \text{ we have } [c, d] \cap A \cap N(x) \neq \varnothing \},
  $$
  and set $\Pi(X, N, A) = (X, N^{\pi})$.
  \end{enumerate}
\end{defi}

\begin{prop}
  The following diagrams commute on objects, up to natural isomorphism.
  $$
    \begin{tikzcd}[column sep=3em]
      \cat{NA}
            \arrow[d, "\sigma"]
            \arrow[r, shift left=2pt, "\fun{uf}"]
        & \cat{DNF}
            \arrow[d, "\Sigma"]
            \arrow[l, shift left=2pt, "\fun{clp}"]
        & \cat{NA}
            \arrow[d, "\pi"]
            \arrow[r, shift left=2pt, "\fun{uf}"]
        & \cat{DNF}
            \arrow[d, "\Pi"]
            \arrow[l, shift left=2pt, "\fun{clp}"] \\ [.5em]
      \cat{CANA}
            \arrow[r, shift left=2pt, "\fun{at}"]
        & \cat{NF}
            \arrow[l, shift left=2pt, "\wp"]
        & \cat{CANA}
            \arrow[r, shift left=2pt, "\fun{at}"]
        & \cat{NF}
            \arrow[l, shift left=2pt, "\wp"]
    \end{tikzcd}
  $$
\end{prop}

\begin{proof}
  Let $(X, N, A)$ be a descriptive neighborhood frame.
  Recall that $\dbox_N a = \{ x \in X \mid a \in N(x) \}$ for $a \in A$.
  Then $(A, \dbox_N)$ is the dual modal algebra of $(X, N, A)$.
  It is well known that the canonical extension of $A$ is $\wp X$.
  Write $(X, N')$ for the neighborhood frame dual to $(A^{\sigma}, \dbox_N^{\sigma})$.
  We aim to show that for all $x \in X$ and $e \subseteq X$ we have
  $e \in N'(x)$ iff $e \in N^{\sigma}(x)$.
  By definition,
  $$
    \dbox_N^{\sigma}e
      = \bigcup \Big\{ \bigcap  \{ \dbox_N a \mid a \in A, c \subseteq a \subseteq d \}
        \mid c \in \fun{K}_A, d \in \fun{O}_A, c \subseteq e \subseteq d \Big\}.
  $$
  Thus, $e \in N'(x)$ iff $x \in \dbox_N^{\sigma}e$, which happens iff there are closed $c$
  and open $d$ such that $c \subseteq e \subseteq d$ and
  $a \in [c, d] \cap A$ implies $x \in \dbox_N a$.
  The latter means that we have $a \in N(x)$ for all such $a$, and hence
  $c$ and $d$ witness the fact that $e \in N^{\sigma}(x)$.
  
  A similar reasoning proves the statement for $\pi$-extensions.
\end{proof}

  The extensions $\Sigma$ and $\Pi$ are closely related.
  To see this, we need the notion of \emph{dual} (descriptive) neighborhood frames.

\begin{defi}\label{def:duals}
  For a neighborhood frame $(X, N)$, define
  $$
    N^c : X \to \fun{PP}X : x \mapsto \{ a \subseteq X \mid a \notin N(x) \}.
  $$
  We call $(X, N)^c := (X, N^c)$ the \emph{complement} of $(X, N)$.
  
  If $(X, N, A)$ is a descriptive neighborhood frame,
  then we define
  $$
    N^c_A : X \to \fun{PP}X : x \mapsto \{ a \in A \mid a \notin N(x) \}.
  $$
  It is easy to see that $(X, N_A^c, A)$ is again a descriptive frame.
  We call $(X, N_A^c, A)$ the \emph{complement} of $(X, N, A)$
  and denote it by $(X, N, A)^{c}$.
\end{defi}

  It is easy to see that $((X, N)^c)^c = (X, N)$. Moreover,
  $(\cdot)^c : \cat{NF} \to \cat{NF}$ defines an involution,
  where $f^c = f$ for a neighborhood morphism $f$.
  Similar statements hold for descriptive frames and the descriptive complement.

\begin{prop}\label{prop:sigma-pi-duals}
  Let $(X, N, A)$ be a descriptive neighborhood frame.
  Then
  $$
    \Pi(X, N, A) = (\Sigma(X, N^c_A, A))^{c}.
  $$
  Consequently, $\Sigma(X, N, A) = (\Pi(X, N^c_A, A))^c$.
\end{prop}
\begin{proof}
  Note that $(\Sigma(X, N^c_A, A))^{c} = (X, ((N^c_A)^{\sigma})^c)$.
  We need to prove that for all $x \in X$ and $e \subseteq X$ we have
  \begin{equation}\label{eq:double-complement}
    e \in N^{\pi}(x) \iff e \in ((N^c_A)^{\sigma})^c(x).
  \end{equation}
  We do so by unravelling the definitions.
  
  For $x \in X$ and $e \subseteq X$ we have
  $e \in ((N^c_A)^{\sigma})^c(x)$ iff $e \notin (N^c_A)^{\sigma}(x)$.
  In other words, $e \in ((N^c_A)^{\sigma})^c(x)$ iff
  we can find no closed $c \in \fun{K}_A$ and open $d \in \fun{O}_A$
  such that $c \subseteq e \subseteq d$ and $[c, d] \cap A \subseteq N^c(x)$.
  Therefore, $e \in ((N^c_A)^{\sigma})^c(x)$ iff
  for all $c \in \fun{K}_A$ and $d \in \fun{O}_A$ such that
  $c \subseteq e \subseteq d$ we have $[c, d] \cap A \cap N(x) \neq \emptyset$.
  But this is exactly the definition of $e \in N^{\pi}(x)$, so
  \eqref{eq:double-complement} holds.
\end{proof}

%--------------------------------------------------------------------------------
\subsection{$\sigma$- and $\pi$-descriptive frames}

  As we have seen in Section~\ref{subsec:problem},
  we do not always have a forgetful functor
  $\fun{U} : \cat{Coalg}(\fun{D}_{\Ax}) \to \cat{Coalg}(\fun{B}_{\Ax})$.
  In particular, recall that convexity and monotonicity are not preserved.
  This is, in part, solved by replacing $\fun{U}$ with either
  $\Sigma$ or $\Pi$.
  
  If $(X, N, A)$ is monotone, then the definition of $N^{\sigma}$ simplifies to
  $$
    N^{\sigma}(x)
      = \{ e \in \fun{P}X \mid \exists c \in \fun{K}_A
        \text{ with } c \subseteq e
        \text{ and } [c, X] \cap A \subseteq N(x) \}.
  $$
  It is then easy to see that the neighborhood frame
  $\Sigma(X, N, A) = (X, N^{\sigma})$ is monotone as well.
  Next suppose $(X, N, A)$ is convex and there are $e, e' \in N^{\sigma}(x)$
  and $e'' \subseteq X$ such that $e \subseteq e'' \subseteq e'$.
  Then by definition of $N^{\sigma}$ we have closed sets $c, c'$ and open sets
  $d, d'$ such that $e \in [c, d]$, $e' \in [c', d']$, $[c, d] \cap A \subseteq N(x)$
  and $[c', d'] \cap A \subseteq N(x)$.
  Since $c \subseteq e \subseteq e'' \subseteq e' \subseteq d'$,
  convexity of $(X, N, A)$ implies $[c, d'] \cap A \subseteq N(x)$.
  This, in turn, witnesses the fact that $e'' \in N^{\sigma}(x)$.
  Therefore, $(X, N^{\sigma})$ is a convex neighborhood frame.
  
  Thus, on objects we have the following well-defined assignments:
  $$
    \Sigma : \cat{Coalg}(\fun{D}_{\axM}) \to \cat{Coalg}(\fun{B}_{\axM})
    \qquad\text{and}\qquad
    \Sigma : \cat{Coalg}(\fun{D}_{\axConv}) \to \cat{Coalg}(\fun{B}_{\axConv}).
  $$
  From the connection between $\Sigma$ and $\Pi$ discussed in
  Proposition~\ref{prop:sigma-pi-duals} we get that
  $$
    \Pi : \cat{Coalg}(\fun{D}_{\axM}) \to \cat{Coalg}(\fun{B}_{\axM})
    \qquad\text{and}\qquad
    \Pi : \cat{Coalg}(\fun{D}_{\axCoConv}) \to \cat{Coalg}(\fun{B}_{\axCoConv})
  $$
  are well-defined assignments as well. Here \axCoConv\ denotes the co-convexity axiom
  \begin{enumerate}
    \axitem{ax:coconv}{CoConv}
          $\Box v \to \Box(v \wedge v') \vee \Box(v \vee v'')$
  \end{enumerate}
  On neighborhood frames $(X, N)$ this corresponds to $N^c(x)$ being convex for all $x \in X$.
  
  We can turn the assignments $\Sigma$ and $\Pi$ into ``proper'' forgetful functors 
  by incorporating the additional neighborhoods that arise from
  $\Sigma$ or $\Pi$ into the notion of a descriptive frame.
  This yields two alternative definitions of descriptive frames:
  $\sigma$-descriptive and $\pi$-descriptive frames.
  In what follows we will focus on $\sigma$-descriptive frames,
  leaving the dual treatment of $\pi$-descriptive frames to the reader.
  
\begin{defi}\label{def:sigma-descriptive}
  A \emph{$\sigma$-descriptive neighborhood frame} is a general neighborhood frame
  $(X, N, A)$ that is differentiated and compact (see Definition~\ref{def:gen})
  and satisfies the following modification of the tightness condition:
  for all $x \in X$ and $e \in \fun{P}X$
  $$
    e \in N(x) \iff \exists c \in \fun{K}_A, d \in \fun{O}_A
               \text{ with } c \subseteq e \subseteq d
               \text{ and } [c, d] \cap A \subseteq N(x).
  $$
  We write $\cat{DNF}^{\sigma}$ for the category of $\sigma$-descriptive
  neighborhood frames and general neighborhood morphisms.
\end{defi}

  If $(X, N, A)$ is a descriptive neighborhood frame, then
  $(X, N^{\sigma}, A)$ is a $\sigma$-descriptive neighborhood frame.
  Conversely, given a $\sigma$-descriptive neighborhood frame $(X, N, A)$,
  setting $N_A(x) = N(x) \cap A$ yields a descriptive neighborhood
  frame $(X, N_A, A)$. It is straightforward to see that these two assignments
  give rise to a bijective correspondence between objects in $\cat{DNF}$ and 
  $\cat{DNF}^{\sigma}$.
  Thus, the following diagram commutes on objects, where $\fun{U}$ is the
  forgetful functor that does not add any neighborhoods:
  $$
    \begin{tikzcd}[column sep=.5em]
      \cat{DNF}
            \arrow[rr, <->, "\text{1-1}"]
            \arrow[dr, "\Sigma" swap]
        &
        & \cat{DNF}^{\sigma}
            \arrow[dl, "\fun{U}"] \\
        & \cat{NF}
        &
    \end{tikzcd}
  $$
  
  While in descriptive neighborhood frames (Definition \ref{def:gen})
  the tightness condition stipulates
  that all neighborhoods are admissible sets, $\sigma$-tightness allows non-admissible sets
  to act as neighborhoods too.

\begin{rem}\label{rem:pi-descr}
  Let $\cat{DNF}^{\pi}$ be the category of $\pi$-descriptive frames,
  defined analogously to Definition~\ref{def:sigma-descriptive}.
  That is, a $\pi$-descriptive frame is a general neighborhood frame that is
  differentiated and compact and satisfies for all $x \in X$ and $e \subseteq X$:
  $$
    e \in N(x) \iff \forall c \in \fun{K}_A, d \in \fun{O}_A
    \text{ with } c \subseteq e \subseteq d
    \text{ we have } [c, d] \cap A \cap N(x) \neq \varnothing.
  $$
  We have that $(X, N, A)$ is $\sigma$-descriptive iff $(X, N^c_A, A)$ is
  $\pi$-descriptive.
  Moreover, 
  $f : (X, N, A) \to (X', N', A')$ is a
  general morphism between $\sigma$-descriptive frames iff
  $f$ is a general morphism between $(X, N_A^c, A)$ and $(X', (N')_A^c, A')$.
  Thus, we obtain an isomorphism
  between $\cat{DNF}^{\sigma}$ and $\cat{DNF}^{\pi}$.
\end{rem}

  Write $\cat{DNF}^{\sigma}(\Ax)$ for the full subcategory of $\cat{DNF}$
  such that the axioms in $\Ax$ are satisfied when interpreting the
  variables used in $\Ax$ as clopens.
  Write also $\fun{U}_{\Ax} : \cat{DNF}^{\sigma}(\Ax) \to \cat{Coalg}(\fun{B}_{\Ax})$
  for the functor that sends a $\sigma$-descriptive neighborhood frame to its
  underlying frame (viewed as a coalgebra).
  If $\fun{U}_{\Ax}$ is well defined, then it is automatically
  a functor, because the additional (non-admissible) neighborhoods of
  $\sigma$-descriptive frames ensure that every morphism in $\cat{DNF}^{\sigma}(\Ax)$
  is in particular a neighborhood morphism between the underlying neighborhood frames.

  While the introduction of $\sigma$-descriptive frames ensures that
  $\fun{U}_{\Ax} : \cat{DNF}^{\sigma}(\Ax) \to \cat{Coalg}(\fun{B}_{\Ax})$
  becomes a functor,
  it only moves the problem of functoriality elsewhere.
  Indeed, we are not guaranteed that the category $\cat{Alg}(\fun{N}_{\Ax})$
  is dual to $\cat{DNF}^{\sigma}(\Ax)$.
  
  In Section \ref{subsec:ce-fun} we will prove that whenever $\Ax$ implies $\axConv$,
  then 
  \begin{enumerate}
    \item $\cat{DNF}^{\sigma}(\Ax)$ is a category of coalgebras for an endofunctor
          $\fun{D}^{\sigma}_{\Ax}$ on $\cat{Stone}$; and
    \item $\fun{N}_{\Ax}$ is dual to $\fun{D}^{\sigma}_{\Ax}$.
  \end{enumerate}
  Combined, these give the dual equivalence
  $
    \cat{Alg}(\fun{N}_{\Ax})
      \equiv^{\op} \cat{Coalg}(\fun{D}^{\sigma}_{\Ax})
      \cong \cat{DNF}^{\sigma}(\Ax)
  $.

%--------------------------------------------------------------------------------
\subsection{When are $\Sigma$ and $\Pi$ functors?}\label{subsec:ce-fun}

  In this section we give a sufficient condition for $\fun{D}^{\sigma}_{\Ax}$ to
  be a functor.

\begin{defi}
  Let $\topo{X}$ be a Stone space with the underlying set $X$. We write
  $\fun{K}\topo{X}$ and $\fun{O}\topo{X}$ for the closed and open sets
  of $\topo{X}$ and define $\fun{D}^{\sigma}\topo{X}$ to be the space consisting of
  $W \subseteq \fun{P}X$ that satisfy
  $$
    e \in W \iff \exists c \in \fun{K}\topo{X}, d \in \fun{O}\topo{X} \text{ with }
             c \subseteq e \subseteq d \text{ and } [c, d] \cap \fun{clp}\topo{X} \subseteq W.
  $$
  The topology on $\fun{D}^{\sigma}\topo{X}$ is generated by the clopen subbase
  $$
    \lbox a = \{ W \mid a \in W \}, \qquad
    \ldiamond a = \{ W \mid X \setminus a \notin W \},
  $$
  where $a$ ranges over $\fun{clp}\topo{X}$.
\end{defi}

  Using the fact that elements of $\fun{D}^{\sigma}\topo{X}$ are determined
  uniquely by the clopens of $\topo{X}$ they contain, combined with assignments
  similar to the ones in the paragraph following Definition~\ref{def:sigma-descriptive},
  it is easy to see that
  $\fun{D}^{\sigma}\topo{X}$ is homeomorphic to $\fun{D}\topo{X}$.
  Therefore, $\fun{D}^{\sigma}\topo{X}$ is a Stone space.
  Moreover, $W \in \fun{D}^{\sigma}\topo{X}$ is an $\Ax$-subset
  (see~Definition \ref{def:ax-subset-stone}) iff $W \cap \fun{clp}\topo{X} \in \fun{D}\topo{X}$ is an $\Ax$-subset.
  Thus, we define:
  
\begin{defi}
  For a Stone space $\topo{X}$, let $\fun{D}^{\sigma}_{\Ax}\topo{X}$ be the
  subspace of $\fun{D}^{\sigma}\topo{X}$ whose elements are $\Ax$-subsets.
  For a continuous function $f : \topo{X} \to \topo{X}'$,
  we define
  $$
    \fun{D}^{\sigma}_{\Ax}f
      : \fun{D}^{\sigma}_{\Ax}\topo{X} \to \fun{D}^{\sigma}_{\Ax}\topo{X}'
      : W \mapsto \{ w' \subseteq \topo{X}' \mid f^{-1}(w') \in W \}.
  $$
\end{defi}

  Since $\fun{D}^{\sigma}_{\Ax}\topo{X}$ is homeomorphic to
  $\fun{D}_{\Ax}\topo{X}$, we have that $\fun{D}^{\sigma}_{\Ax}$
  sends a Stone space to a Stone space.
  Furthermore, this implies that $\fun{D}^{\sigma}_{\Ax}$ is
  naturally isomorphic to $\fun{D}_{\Ax}$
  \emph{whenever the former is well defined}.
  We prove that it is well defined when $\Ax$ implies $\axConv$.

\begin{thm}\label{thm:D-Dsigma}
  Let $\Ax$ be a set of finitary one-step axioms such that
  $\Ax$ implies \axConv.
  Then $\fun{D}^{\sigma}_{\Ax}$ defines an endofunctor on $\cat{Stone}$.
\end{thm}

\begin{proof}
  We have already seen that $\fun{D}^{\sigma}_{\Ax}$ is well defined on
  objects, and functoriality is straightforward because it is a subfunctor
  of $\fun{B}$.
  Therefore, we only have to prove that for every continuous function
  $f : \topo{X} \to \topo{X}'$, the assignment
  $\fun{D}^{\sigma}_{\Ax}f : \fun{D}^{\sigma}_{\Ax}\topo{X} \to \fun{D}^{\sigma}_{\Ax}\topo{X}'$
  is a well-defined continuous function.
  
  Continuity follows from the fact that
  $(\fun{D}^{\sigma}_{\Ax}f)^{-1}(\lbox a') = \lbox f^{-1}(a')$ for all
  $a' \in \fun{clp}\topo{X}'$.
  So all that is left is to prove that
  $\fun{D}^{\sigma}_{\Ax}f(W) \in \fun{D}^{\sigma}_{\Ax}\topo{X}'$ for all
  $W \in \fun{D}^{\sigma}_{\Ax}\topo{X}$.
  That is, we need to show that for all $e' \subseteq X'$, 
  \begin{equation}\label{eq:Dsigax-goal}
  \begin{split}
    e' \in \fun{D}^{\sigma}_{\Ax}f(W)
      \iff &\exists c' \in \fun{K}\topo{X}', d' \in \fun{O}\topo{X}' \\
      &\text{with } c' \subseteq e' \subseteq d'
      \text{ and } [c', d'] \cap \fun{clp}\topo{X}' \subseteq \fun{D}^{\sigma}_{\Ax}f(W).
  \end{split}
  \end{equation}
  
  First assume $e' \in \fun{D}^{\sigma}_{\Ax}f(W)$.
  Then $f^{-1}(e') \in W$ and so there exist $c \in \fun{K}\topo{X}$
  and $d \in \fun{O}\topo{X}$ such that $c \subseteq f^{-1}(e) \subseteq d$
  and $[c, d] \cap \fun{clp}\topo{X} \subseteq W$.
  Define $c' = f[c]$ and $d' = \topo{X}' \setminus f[\topo{X} \setminus d]$.
  Since $f$ is a continuous function between Stone spaces, it sends closed sets
  to closed sets, so $c' \in \fun{K}\topo{X}'$ and $d' \in \fun{O}\topo{X}'$.
  Furthermore, we claim that $c' \subseteq e' \subseteq d'$.
  The first inclusion is obvious. For the second, if
  $x' \notin d'$, then there is $x \in \topo{X} \setminus d$ such that
  $f(x) = x'$. But then $x \notin d$, so $x \notin f^{-1}(e')$, and hence $x' = f(x) \notin e'$.
  
  We claim that $c'$ and $d'$ witness that the right-hand side of
  \eqref{eq:Dsigax-goal} holds.
  Let $a' \in \fun{clp}\topo{X}'$ such that $c' \subseteq a' \subseteq d'$.
  Then $c \subseteq f^{-1}(a')$ by definition of $c'$.
  Furthermore, $f^{-1}(d') \subseteq d$.
  To see this, $x \notin d$ implies $x \in \topo{X} \setminus d$,
    so $f(x) \in f[\topo{X} \setminus d]$. Therefore,
    $f(x)\notin \topo{X}' \setminus f[\topo{X}\setminus d] = d'$,
    and hence $x \notin f^{-1}(d')$.
  Thus, $f^{-1}(a') \subseteq d$.
  By assumption this implies that $f^{-1}(a') \in W$. Consequently, 
  $a' \in \fun{D}^{\sigma}_{\Ax}f(W)$.
  
  For the converse, suppose $e' \subseteq X'$ is such that the right-hand side holds.
  Denote the relevant closed and open subsets of $\topo{X}'$ witnessing this
  by $c'$ and $d'$.
  We aim to show that $f^{-1}(e') \in W$.
  To prove this, it suffices to show that there exist a closed and open subsets
  $c$ and $d$ of $\topo{X}$ such that $c \subseteq f^{-1}(e') \subseteq d$
  and $[c, d] \cap \fun{clp}\topo{X} \subseteq W$.
  
  Take $c = f^{-1}(c')$ and $d = f^{-1}(d')$. By continuity of $f$ these are
  closed and open, respectively. Now let $a \in \fun{clp}\topo{X}$ be such that
  $c \subseteq a \subseteq d$.
  To prove that $a \in W$, we construct $b_1', b_2' \in \fun{clp}\topo{X}'$
  such that $b_1', b_2' \in \fun{D}^{\sigma}_{\Ax}f(W)$ and
  $f^{-1}(b_1') \subseteq a \subseteq f^{-1}(b_2')$.
  Convexity of $W$ then implies that $a \in W$.
  
  To construct $b_1'$, since $c'$ is closed, we have that
  $c' = \bigcap \{ k' \in \fun{clp}\topo{X}' \mid c' \subseteq k' \}$.
  Therefore,
  $$
    c = f^{-1}(c') = \bigcap \big\{ f^{-1}(k') \mid k' \in \fun{clp}\topo{X}' \text{ and } c' \subseteq k' \big\} \subseteq a.
  $$
  Since $a$ is clopen and the intersection is directed, we can find
  $k_1' \in \fun{clp}\topo{X}'$ such that $c \subseteq f^{-1}(k_1') \subseteq a$.
  Similarly, since $c' = \bigcap \{ k' \in \fun{clp}\topo{X}' \mid c' \subseteq k' \} \subseteq d'$ we can find
  $k_2' \in \fun{clp}\topo{X}'$ such that
  $c' \subseteq k_2' \subseteq d'$.
  Setting $b_1' = k_1' \cap k_2'$ gives an element in
  $\fun{clp}\topo{X}$ such that $f^{-1}(b_1') \subseteq a$.
  Moreover, by construction $c' \subseteq b_1' \subseteq d'$,
  and hence $b_1' \in \fun{D}^{\sigma}_{\Ax}f(W)$.
  
  Finally, to construct $b_2'$, since $d'$ is open, we have
  $d' = \bigcup \{ k' \in \fun{clp}\topo{X}' \mid k' \subseteq d' \}$.
  A similar argument to the above yields $b_2'$ satisfying all the required properties.
  This proves that $a \in W$, hence $f^{-1}(e') \in W$, completing the proof.
\end{proof}

  The definition of descriptive monotone frames from \cite{HanKup04}
  coincides with the definition of $\sigma$-descriptive neighborhood frames
  satisfying monotonicity. Let us write $\cat{HDMF}$ for the category
  of such descriptive monotone frames, as defined in Section 2.4.2 of {\it op.~\!cit}.
  The algebraic semantics of monotone modal logic is given by
  \emph{monotone Boolean algebra expansions} (BAMs),
  and we write $\cat{BAM}$ for the full subcategory of $\cat{NA}$ whose objects
  are BAMs \cite[Section 2.4.1]{HanKup04}.
  As a consequence of Theorem \ref{thm:D-Dsigma} we
  now obtain the duality from \cite[Theorem 2.11]{HanKup04}
  as an algebra/coalgebra duality.

\begin{cor}[Hansen-Kupke]
  $\cat{BAM} \equiv^{\op} \cat{HDMF}$.
\end{cor}
\begin{proof}
  One can prove that $\cat{BAM} \cong \cat{Alg}(\fun{N}_{\axM})$ in the same way
  as in Proposition \ref{prop:alg}\eqref{it:prop:alg-N}.
  Moreover, $\cat{HDMF} \cong \cat{Coalg}(\fun{UpV})$ \cite[Theorem 3.12]{HanKup04}, where
  $\fun{UpV}$ is the endofunctor on $\cat{Stone}$ defined in
  Definition 3.9 of {\it op.~\!cit.}
  There is a natural isomorphism $\nu : \fun{D}^{\sigma}_{\axM} \to \fun{UpV}$
  given on components by
  $$
    \nu_{\topo{X}} : \fun{D}^{\sigma}_{\axM}\topo{X} \to \fun{UpV}\topo{X} : W \mapsto \{ c \in W \mid c \text{ is closed in } \topo{X} \}.
  $$
  This is injective because if $W, V \in \fun{D}^{\sigma}_{\axM}$ are distinct,
  then there exists $a \in \fun{clp}\topo{X}$ such that $a \in W$ and $a \notin V$,
  or $a \notin W$ and $a \in V$. Since clopen sets are in particular closed,
  this implies that $\nu_{\topo{X}}(W) \neq \nu_{\topo{X}}(V)$, so $\nu_{\topo{X}}$ is injective.
  Moreover, it is surjective. To see this, observe that for all $W \in \fun{UpV}\topo{X}$
  the set $W^{\uparrow} = \{ e \subseteq \topo{X} \mid \exists c \in W \text{ such that } c \subseteq e \}$
  is in $\fun{D}^{\sigma}_{\axM}\topo{X}$ and satisfies $\nu_{\topo{X}}(W^{\uparrow}) = W$.
  So $\nu$ is a bijective continuous function, hence a homeomorphism.

  Combining Theorems~\ref{thm:fun-dual-stone}
  and~\ref{thm:D-Dsigma} yields that $\fun{N}_{\axM}$ is dual
  to $\fun{UpV}$. Together with known isomorphisms this yields
  $$
    \cat{BAM}
      \cong \cat{Alg}(\fun{N}_{\axM})
      \equiv^{\op} \cat{Coalg}(\fun{UpV})
      \cong \cat{HDMF},
  $$
  which proves the desired duality.
\end{proof}

\begin{cor}
  Let $\Ax$ be a collection of finitary one-step axioms that implies \axConv.
  Then
  $$
    \cat{Alg}(\fun{N}_{\Ax}) \equiv^{\op} \cat{Coalg}(\fun{D}^{\sigma}_{\Ax}).
  $$
\end{cor}
\begin{proof}
  Combine Corollary~\ref{cor:fun-dual-stone}, Theorem~\ref{thm:D-Dsigma},
  and the fact that $\fun{D}_{\Ax}$ is naturally isomorphic to $\fun{D}^{\sigma}_{\Ax}$
  to obtain 
  \[
  \cat{Alg}(\fun{N}_{\Ax}) \equiv^{\op} \cat{Coalg}(\fun{D}_{\Ax}) \cong \cat{Coalg}(\fun{D}^{\sigma}_{\Ax}). \qedhere
  \]
\end{proof}

\begin{cor}
  Let $\Ax$ be a collection of finitary one-step axioms that implies \axConv.
  If the restriction $\Sigma_{\Ax}$ of $\Sigma$ to $\cat{Coalg}(\fun{D}_{\Ax})$
  lands in $\cat{Coalg}(\fun{B}_{\Ax})$,
  then $\Sigma_{\Ax}$ is a functor.
\end{cor}
\begin{proof}
  If $\Sigma_{\Ax}$ is well defined, then so is $\fun{U}_{\Ax}$, and
  hence $\fun{U}_{\Ax}$ is a functor (as we pointed out after Remark \ref{rem:pi-descr}).
  Since $\Sigma_{\Ax}$ can be obtained as the composition of functors
  $$
    \begin{tikzcd}
      \cat{Coalg}(\fun{D}_{\Ax})
            \arrow[r, "\cong"]
            \arrow[rr, bend right=20, "\Sigma_{\Ax}" swap]
        & \cat{Coalg}(\fun{D}_{\Ax}^\sigma)
            \arrow[r, "\fun{U}_{\Ax}"]
        & \cat{Coalg}(\fun{B}_{\Ax})
    \end{tikzcd}
  $$
  it is a functor as well.
\end{proof}

  We can now obtain relatively easily the commuting diagrams relating
  the J\'{o}nsson-Tarski type and Thomason type dualities for
  a large class of modal logics.
  These are analogues of the diagrams for basic normal modal logic depicted
  in~\eqref{eq:square-normal}.
  We formulate this as a general statement, and then instantiate it to
  basic monotone modal logic.
  Observe that this still requires a preservation result reminiscent
  of a Sahlqvist theorem, proving that validity of axioms on a $\sigma$-descriptive neighborhood frame implies validity of the axioms on the underlying neighborhood frame (see, e.g.,~\cite{SamVac89} or \cite[Section 5.6]{BRV01}). We leave the search for such theorems
  to future research.

\begin{thm}\label{thm:overview}
  Let $\Ax$ be a collection of finitary one-step axioms that implies \axConv.
  Suppose that for every Stone space $\topo{X}$,
  every $\Ax$-neighborhood $W$ of $\fun{D}\topo{X}$ is also an
  $\Ax$-neighborhood of $\fun{B}X$, where $X$ is the set underlying $\topo{X}$.
  Then the following diagram commutes
  $$
    \begin{tikzcd}
      \cat{Alg}(\fun{N}_{\Ax})
            \arrow[r, -, "\equiv^{\op}"]
            \arrow[d, "\sigma" swap]
        & \cat{Coalg}(\fun{D}_{\Ax})
            \arrow[d, "\Sigma_{\Ax}"]
            \arrow[r, -, "\cong"]
        & \cat{Coalg}(\fun{D}^{\sigma}_{\Ax})
            \arrow[dl, bend left=20, "\fun{U}_{\Ax}"] \\
      \cat{Alg}(\fun{L}_{\Ax})
            \arrow[r, -, "\equiv^{\op}"]
        & \cat{Coalg}(\fun{B}_{\Ax})
        &
    \end{tikzcd}
  $$
\end{thm}

\begin{exa}
  As we have seen, examples of such $\Ax$ are $\Ax = \{ \axConv \}$
  and $\Ax = \{ \axM \}$.
  The latter yields the following commuting diagrams, where all
  edges are functors:
  $$
    \begin{tikzcd}[column sep=1.5em]
      \cat{Alg}(\fun{N}_{\axM})
            \arrow[r, -, "\equiv^{\op}"]
            \arrow[d, "\sigma" swap]
        & [.5em]
          \cat{Coalg}(\fun{D}_{\axM})
            \arrow[r, -, "\cong"]
            \arrow[d, "\Sigma_{\axM}"]
        & \cat{Coalg}(\fun{D}^{\sigma}_{\axM})
            \arrow[ld, bend left=20, "\fun{U}_{\axM}"]
        & [-.5em] 
          \cat{BAM}
            \arrow[r, -, "\equiv^{\op}"]
            \arrow[d, "\sigma" swap]
        & \cat{Coalg}(\fun{D}_{\axM})
            \arrow[r, -, "\cong"]
            \arrow[d, "\Sigma_{\axM}"]
        & \cat{HDMF}
            \arrow[dl, bend left=20, "\fun{U}"] \\
      \cat{Alg}(\fun{L}_{\axM})
            \arrow[r, -, "\equiv^{\op}"]
        & \cat{Coalg}(\fun{B}_{\axM})
        &
        & \cat{CABAM}
            \arrow[r, -, "\equiv^{\op}"]
        & \cat{MF}
        &
    \end{tikzcd}
  $$
\end{exa}

%================================================================================
\section{Conclusions}

  We have given a general coalgebraic approach to Thomason type dualities for
  neighborhood frames and J\'{o}nsson-Tarski type dualities for neighborhood
  algebras. Furthermore, we have investigated the relationship between the two
  types of dualities via the theory of canonical extensions.
  We list several potential avenues for future research.
  
   \vspace{2mm}
   
  \begin{description}[itemsep=2mm]
    \item[Infinitary modal logic] 
          The Thomason type dualities from Sections \ref{sec:thomason}
          and \ref{sec:harvest} provide dualities for algebraic and geometric
          semantics for \emph{infinitary} modal logic. While some interesting
          investigations have been conducted by Baltag \cite{Baltag98PhD,Baltag98},
          obtaining a more general coalgebraic approach towards infinitary modal
          logic (also using the results of this  paper) remains open. 

    \item[Endofunctors as left adjoints]
          In \cite{BezCarMor20} the functor $\fun{H}$ is obtained as the
          left adjoint of the forgetful functor from $\cat{CABA}$ to
          $\cat{CSL}$---the category of complete meet-semilattices.
          In a similar way we can obtain the functor whose coalgebras are
          monotone neighborhood frames as the left adjoint of the forgetful
          functor $\cat{CABA} \to \cat{Pos}$,
          and the functor whose coalgebras are filter frames arises as the
          left adjoint of $\cat{CABA} \to \cat{SL}$, where $\cat{SL}$ is
          the category of (not necessarily complete) meet-semilattices.
          An interesting direction for future work is to investigate the
          connection between presentations of classes of algebras via adjoints
          of forgetful functors and via one-step axioms of infinitary logic.
 
    \item[When are $\Sigma$ and $\Pi$ functors? A different approach]
          In Section \ref{subsec:ce-fun} we proved that in presence of
          the convexity axiom, $\Sigma$ defines a functor.
          Another approach towards the same goal is by modifying the morphisms
          between (descriptive) neighborhood frames.
          If we define a $\sigma$-morphism to be a function
          between neighborhood frames that satisfies only the left-to-right
          implication in \eqref{eq:nbhd-mor},
          then replacing neighborhood morphisms with $\sigma$-morphisms ensures
          that $\Sigma$ is a functor without adding any additional axioms.
          Considering such morphisms is natural for a number of reasons.
          For example, in the special case of topological spaces they simply
          correspond to continuous maps. They also generalize the stable
          morphisms studied in \cite{BBI16, BBI18} to the setting of (descriptive)
          neighborhood frames. We plan to investigate this topic further in a
          sequel to this paper. 
  \end{description}

\section*{Acknowledgements}
  We are grateful to the referees for careful reading and
  useful comments that have improved the presentation of the paper.

%================================================================================
\bibliographystyle{alphaurl}
\bibliography{n-frm-biblio.bib}

\end{document}